\documentclass{CRPITStyle} 
\usepackage[authoryear]{natbib}
\renewcommand{\cite}{\citep}
\def \trans#1	{\stackrel{#1}{\longrightarrow}}
\def \supp		{\mathrm{supp}}
\def \pr#1		{\empty_{#1}\!\| }
\def \pc#1		{\empty_{#1}\!\!\oplus}
\def \flip		{\mathtt{flip}}
\def \ars 		{\ar@{.>}}
\def \praut		{p\textrm{-}\mathbf{Aut}}
\def \arrs 	{\ar@{<.}}
\def \arss		{\ar@{<.>}}
\def \eps		{\varepsilon}
\def \refby 	{\sqsubseteq}
\def \aut 		{\mathbf{Aut}}
\def \paut 		{\mathbf{ProbAut}}
\def \may#1 	{\refby_{\texttt{\tiny{may}}}^{#1}}
\def \tea		{\mathtt{tea}}
\def \coffee		{\mathtt{coffee}}
\def \coin		{\mathtt{coin}}
\def \lra		{\longrightarrow}
\def \D 		{\mathcal{D}}
\def \ttrans#1	{\stackrel{#1}{\Longrightarrow}}
\def \rf#1 {(\ref{#1})}
\usepackage{xy}
\xyoption{all}
\usepackage{amsthm}
\usepackage{amsfonts}
\theoremstyle{plain}
\newtheorem{theorem}{Theorem}
\newtheorem{proposition}[theorem]{Proposition}
\newtheorem{definition}[theorem]{Definition}
\newtheorem{lemma}[theorem]{Lemma}
\newtheorem{remark}[theorem]{Remark}

\usepackage{hyperref}
\usepackage{graphicx}

\begin{document}

\title{Weak Concurrent Kleene Algebra with Application to Algebraic Verification}
\author{Annabelle McIver$^1$
\and 
Tahiry Rabehaja$^{1}$
\and 
Georg Struth$^{3}$}
\affiliation{$^1$ Department of Computing \\
Macquarie University, \\
Sydney, Australia \\
Email:~{\tt \{annabelle.mciver,tahiry.rabehaja\}@mq.edu.au}\\[.1in]
$^2$ Department of Computer Science \\
University of Sheffield, \\
United Kingdom, \\
Email:~{\tt g.struth@dcs.shef.ac.uk}}

\maketitle

\toappear{This work has been supported by the iMQRES grant from Macquarie University.}





\begin{abstract}
We propose a generalisation of concurrent Kleene algebra \cite{Hoa09} that can take account of probabilistic effects in the presence of concurrency. The algebra is proved sound with respect to a model of automata modulo a variant of rooted $\eta$-simulation equivalence. Applicability  is demonstrated by algebraic treatments of two examples: algebraic may testing and Rabin's solution to the choice coordination problem.
 \end{abstract}

\section{Introduction}
Kleene algebra generalises the language of regular expressions and, as a basis for reasoning about programs and computing systems, it has been used in applications ranging from compiler optimisation, program refinement, combinatorial optimisation and algorithm design~\cite{Con71,Koz94,Koz00a,Koz00b,Mci06}. A number of variants of the original axiom system and language of Kleene algebra have extended its range of applicability to include probability \cite{Mci05} with the most recent being the introduction of a concurrency operator \cite{Hoa09}. Main benefits of the algebraic approach are that it captures some essential aspects of computing systems in a simple and concise way and that the calculational style of reasoning it supports is very suitable for automated theorem proving. 

In this paper we continue this line of work and propose \emph{weak concurrent Kleene algebra}, which extends the abstract probabilistic Kleene algebra~\cite{Mci05} with the concurrency operator of concurrent Kleene algebra~\cite{Hoa09} and thus supports reasoning about concurrency in a context of probabilistic effects. This extension calls for a careful evaluation of the axiom system so that it accurately accounts for the interactions of probabilistic choice, nondeterministic choice and the treatment of concurrency.  For example probabilistic Kleene algebra accounts for the presence of probability in the \emph{failure} of the original distributive law $x(y + z) = xy + xz$ which is also absent in most process algebras. That is because when the terms $x, y, z$ are interpreted as probabilistic programs, with $xy$ meaning ``first execute $x$ and then $y$" and $+$ interpreted as a nondeterministic choice, the expression on the left hand side exhibits a greater range of nondeterminism than the right in the case that $x$ includes probabilistic behaviours.  For example if $x$ is interpreted as a program which flips a bit with probability $1/2$ then the following nondeterministic choice in $y+z$ can always be resolved so that $y$ is executed \emph{if and only if} the bit was indeed flipped. This is not a behaviour amongst those described by $xy + xz$, where the nondeterminism is resolved before the bit is flipped and therefore its  resolution  is unavoidably independent of the flipping.
Instead, in contexts such as these, distributivity be replaced by a weaker law:
\begin{equation}\label{eq:subdistributivity}
\textrm{Sub-distributivity:} \qquad xy + xz ~\leq ~ x(y + z)~.
\end{equation}
Elsewhere~\cite{Rab11} we show that this weakening of the original axioms of Kleene algebra results in a complete system relative to  a model of nondeterministic automata modulo simulation equivalence.  

 The behaviour of the concurrency operator of concurrent Kleene algebra~\cite{Hoa09} is captured in particular by the \emph{Interchange law}:
$$
(x \| y) (u \| v) ~ \leq ~ (x u) \| (y v)
$$
which expresses that there is a lesser range of nondeterministic executions on the left where, for example, the execution of $u$ is constrained to follow a complete execution of $x$ run concurrently with $y$ but on the right it is not.

Our {\bf first contribution} is the construction of a concrete model of abstract probabilistic automata (where the probability is at the action level) over which to interpret  terms composed of traditional Kleene algebra  together with concurrent composition. In this interpretation, each term represents an automaton. For example in Equation (\ref{eq:subdistributivity}), $x,y$ and $z$ are automata and so is $xy + xz$. We show that the axiom system of concurrent Kleene algebra weakened to allow for the presence of probability is sound with respect to those probabilistic automata. Our use of probabilistic automata is similar to models where the resolution of probability and nondeterminism can be interleaved; concurrent composition of automata models CSP synchronisation~\cite{Hoa78} in that context. Finally we use a notion of rooted $\eta$-simulation to interpret the inequality $\leq$ used in  algebraic inequations.

Our {\bf second contribution} is to explore some applications of our axiomatisation of weak concurrent Kleene algebra, to explain our definition of   rooted $\eta$-simulation in terms of may testing~\cite{Nic83}, and to demonstrate the proof system on Rabin's distributed consensus protocol~\cite{Rab82}.

One of the outcomes of this study is to expose the tensions between the various aspects of system execution. 
Some of the original concurrent Kleene algebra axioms~\cite{Hoa09} required for the concurrency operator now fail to be satisfiable in the presence of probabilistic effects and synchronisation supported by the interchange law. For example, the term $1$ from Kleene algebra (interpreted as ``do nothing") can no longer be a neutral element for the concurrency operator $\|$ --- we only have the specific equality $1 \| 1 = 1$ and not the more general $1 \| x = x$.
In fact we chose to preserve the full interchange law in our choice of axioms because it captures so many notions of  concurrency already including exact parallel and synchronisation, suggesting that it is a property about general concurrent interactions.

A feature of our approach is to concentrate on broad algebraic structures in order to understand how various behaviours interact rather than to study precise quantitative behaviours. Thus we do not include an explicit probabilistic choice operator in the signature of the algebra --- probability occurs explicitly only  in the concrete model as a special kind of asynchronous probabilistic action combined with internal events (events that the environment cannot access). This allows the specification of complex concurrent behaviour to be simplified using applications of weak distributivity embodied by Equation (\ref{eq:subdistributivity}) and/or the interchange law as illustrated by our case study. Finally we note that the axiomatisation we give is entirely in terms  of first-order expressions and therefore is supported by first-order reasoning. Thus all of our algebraic proofs has been implemented within the Isabelle/HOL theorem proving environment. These proof can be found in a repository of formalised algebraic theorems.~\footnote{\url{http://staffwww.dcs.shef.ac.uk/people/G.Struth/isa/}}

In Section \ref{sec:axiomatisation} we explore the axiomatisation of the new algebra. It is essentially a mixture of probabilistic and concurrent Kleene algebras. Sections \ref{sec:concrete-model} and \ref{sec:soundness} are devoted to showing the consistency of our approach. A concrete model based on automata and $\eta$-simulation is constructed. In section \ref{sec:probabilistic-aut}, we compare our approach with probabilistic automata (automata that exhibit explicit probability) and probabilistic simulation. We conclude that, up to some constraint, the concrete model is a very special case of that more general model. In sections \ref{sec:algebraic-testing} and \ref{sec:rabin-protocol}, we present some applications, in particular an algebraic version of may testing is studied and variations of the specification of Rabin's protocol are explored.

In this paper $x,y,$ etc represent algebraic expressions or variables. Terms are denoted $s,t,$ etc. Letters $a,b,$ etc stand for actions and $\tau$ represents an internal action. An automaton associated with a term or an expression is usually denoted by the same letter. Other notation is introduced as we need it.

In this extended abstract we can only explain the main properties of weak concurrent Kleene algebra and sketch the construction of   the automaton model. Detailed constructions and proofs of all statements in this paper can be found in an appendix.

\section{Axiomatisation}\label{sec:axiomatisation}

A Kleene algebra is a structure that encodes algebraically the sequential behaviour of a system. It is generally presented in the form of an idempotent~\footnote{Idempotence refers to the operation $+$ i.e. $x + x = x$.} semiring structure $(K,+,\cdot,0,1)$ where $x\cdot y$ (sequential composition) is sometimes written using juxtaposition $xy$ in expressions. The term $0$ is the neutral element of $+$ and $1$ is the neutral element of $\cdot$. The semiring is then endowed with a unary Kleene star $*$ representing finite iteration to form a Kleene algebra. This operator is restricted by the following axioms:
\begin{eqnarray}
\textrm{Left unfold:}\qquad 1 + xx^* & = & x^*, \label{eq:unfold}\\
\textrm{Left induction:}\hspace{1.5mm}\qquad xy\leq y&\Rightarrow &x^*y\leq y,\label{hf:linduction}
\end{eqnarray}
where $x\leq y$ if and only if $x + y = y$. In the sequel our interpretations will be over a version of probabilistic automata. In particular we will interpret $\leq$ and $=$ as $\eta$-simulations.

Often, the dual of (\ref{eq:unfold}-\ref{hf:linduction}) i.e. $1 + x^*x = x^*$ and $yx\leq y\Rightarrow yx^*\leq y$ are also required. However, (\ref{eq:unfold}) and (\ref{hf:linduction}) are sufficient here and the dual laws follow from continuity of sequential composition for finite automata.

In a Kleene algebra, the semiring structure supports two distributivity laws:
\begin{eqnarray}
\textrm{Left distributivity:}\qquad\hspace{0.8mm} xy + xz & = & x(y + z), \label{eq:ldist}\\
\textrm{Right distributivity:}\qquad (x + y)z & = & xz + yz. \label{eq:rdist}
\end{eqnarray}  
Equation (\ref{eq:ldist}) however is not valid in the presence of probability. For example, compare the behaviour of probabilistic choice in the diagrams below. Here, $\flip_p$ denotes the process that flips a $p$-biased coin, which we can represent by a probabilistic automaton (details are given in Section \ref{sec:concrete-model}). 
\begin{figure}[h]\label{fig:dist}
$$
\xymatrix{
 	&\ar[dl]_{\flip_p}\ar[dr]^{\flip_p}&	&\hspace{0.5cm}&&\ar[d]_{\flip_p}&\\
\ar[d]_x	&&\ar[d]^y											&&&\ar[dl]_x\ar[dr]^y &\\	
	&&											&& &&
}
$$
\end{figure}
In the right diagram, the choice between $a$ and $b$ can be based on the outcome of the coin flip but such resolution is not possible in the left-hand diagram. We express the greater range of possible outcomes by the general inequation (\ref{eq:subdistributivity}), specifically here it becomes
\begin{equation}\label{eq:lsubdist}
 (\flip_p)y + (\flip_p)z \leq  (\flip_p)(y + z).~\footnote{We have abused notation in this example by using $\flip_p$ to represent both an action and an automaton which performs that action.}
\end{equation}

As mentioned above, the zero of a Kleene algebra satisfies:
\begin{eqnarray}
\textrm{Left annihilation:}\qquad 0x & = & 0, \label{eq:lzero}\\
\textrm{Right annihilation:}\qquad x0 & = & 0. \label{eq:rzero}
\end{eqnarray}

In our interpretation that includes concurrency, we assume that $0$ captures \emph{deadlock}. However, axiom (\ref{eq:rzero}) is no longer appropriate because we should be able to differentiate between the process doing an action and deadlocking from a process that is just deadlocked.

\begin{definition}
A weak probabilistic Kleene algebra is a structure $(K,+,\cdot,*,0,1)$ that satisfies the axioms of Kleene algebra except there is no left distributivity (it is replaced by (\ref{eq:subdistributivity})) and Equation (\ref{eq:rzero}) does not hold generally.
\end{definition}

A concurrency operator was added to Kleene algebra by Hoare et al~\cite{Hoa09}. Our concurrency operator $\|$ satisfies the following standard axioms:
\begin{eqnarray}
\textrm{Associativity:}\qquad x \| (y \| z) & = & (x \| y) \| z, \label{eq:par-assoc}\\
\textrm{Commutativity:}\hspace{1.3cm} x \| y & = & y \| x, \label{eq:par-comm}\\
\textrm{One-idempotence:}\hspace{1.35cm} 1 \| 1 & = & 1.\label{eq:par-1}
\end{eqnarray}

In~\cite{Hoa09}, $\|$ satisfies the identity $1\|x = x$ which we do not have here because in the concrete model, we will interpret $\|$ as the synchronisation operator found in CSP~\cite{Hoa78}. However, we still maintain the instance of that law in the special case $x = 1$ (see axiom~(\ref{eq:par-1})) where $1$ is interpreted as ``do nothing".

Next we have the axioms dealing the interaction of $\|, +$ and $\cdot$.
\begin{eqnarray}
\textrm{Monotonicity :}\qquad x \| y + x \| z & \leq & x \| (y + z) \label{eq:par-dist}\\
\textrm{Interchange-law:}\hspace{0.3cm} (x \| y) (u \| v) & \leq & (x u) \| (y v)\label{eq:exchange-law}
\end{eqnarray}

The interchange law is the most interesting axiom of concurrent Kleene algebra. In fact it allows the derivation of many properties involving $\|$. To illustrate this in the probabilistic context, consider a probabilistic vending machine $\mathtt{VM}$ which we describe as the expression
$$\mathtt{VM}\ =\ \coin\cdot\flip_p\cdot(\tau_h\cdot(\tea+1) + \tau_t\cdot(\coffee+1))$$
where $\coin,\tea,\coffee,\tau_h,\tau_t$ and $\flip_p$ are all represented by automata. 
That is the vending machine accepts a coin and then decides internally whether it will enable the button coffee or tea. The decision is determined by the action $\flip_p$~\footnote{i.e. the automaton that performs a $\flip_p$ action.} which (as explained later) enables either $\tau_h$ or $\tau_t$. The actions $\tau_t$ and $\tau_h$ are internal and the user cannot access them. Now, a user who wants to drink tea is specified as 
$$\mathtt{U}\ =\ \coin\cdot(\tea+1).$$
The system becomes $\mathtt{U}\|\mathtt{VM}$ where the concurrent operation is CSP like and synchronises on $\coin,\tea$ and $\coffee$. The interchange law  together with  the other axioms and some system assumptions imply the following inequation:
\begin{equation}\label{eq:vm}
\mathtt{U}\|\mathtt{VM}\ \geq\  \coin\cdot\flip_p\cdot(\tau_h\cdot(\tea+1) + \tau_t)
\end{equation}
which is proved automatically in our repository. In other words, the user will only be satisfied with \textit{probability at least} $p$ since the right-hand side equation says that the tea action can only be enabled provided that $\tau_h$ is enabled, and in turn that is determined by the result of the $\flip_p$ action.


Now we are ready to define our algebra.

\begin{definition}
A weak concurrent Kleene algebra is a weak probabilistic Kleene algebra $(K,+,\cdot,*,0,1)$ with a concurrency operator $\|$ satisfying (\ref{eq:par-assoc}-\ref{eq:exchange-law})
\end{definition}

We assume the operators precedence $*<\cdot<\|<+$.



%
%
%
%

\begin{proposition}\label{pro:elementary-consequences}
Let $s,t$ be terms, the following equations holds in weak concurrent Kleene algebra.
\begin{enumerate}
\item All the operators are monotonic.
\item $(s^*\|t^*)^* = s^*\|t^*$.\label{eq:star-idem}
\item $(s\|t)^*\leq s^*\| t^*$.\label{eq:subdist}
\item $(s + t)^* = (s^*t^*)^*$.
\end{enumerate}
\end{proposition}


\section{Concrete Model}\label{sec:concrete-model}
\subsection{Semantic Space}\label{subsec:semantic-space}

We use nondeterministic automata to construct a concrete model. An automaton is denoted by a tuple 
\begin{displaymath}
(P,\lra,i,F)
\end{displaymath}
where $P$ is a set of states. The set $\lra\subseteq P\times\Sigma\times P$ is a transition relation and we write $x\trans{a } y$ when there is a transition, labelled by $a$, from state $x$ to state $y$. The alphabet $\Sigma$ is left implicit and considered to be fixed for every automaton. The state $i\in P$ is the initial state and $F\subseteq P$ is the set of final states of the automaton. In the sequel, we will denote an automaton $(P,\lra,i,F)$ by its set of states $P$ when no confusion is possible.

The actions in the alphabet $\Sigma$ are categorised into three kinds:
\begin{itemize}
\item \textit{internal}: actions that will be ``ignored" by the simulation relation (as in $\tau_h$ and $\tau_t$). Internal actions are never synchronised by $\|$.
\item \textit{external}: actions that \emph{can} be synchronised. Probabilistic actions are external (as in $\flip_p$) but they are \emph{never} synchronised.
\item \textit{synchronised}: external actions that will be synchronised when applying $\|$ (as in $\coin,\tea$ and $\coffee$). These actions are determined by a set of external actions $A$. More specifically, $\|$ refers to $\pr{A} $ which we assume is fixed and given beforehand.
\end{itemize}

The special case of probabilistic choice is modelled by combining probabilistic and internal actions. That is a process that does $a$ with probability $p$ and does $b$ with probability $1-p$ is interpreted as the following automaton 
\begin{figure}[h]
$$
\xymatrix{
& \ar[d]^{\flip_p} & \\
& \ar[dl]_{\tau_h}\ar[dr]^{\tau_l} &\\
\ar[d]_{a}&&\ar[d]^{b}\\
& &
}$$
\end{figure}
where $\flip_p\in\Sigma$ represents the action of flipping a $p$-biased coin which produces head with probability $p$ and tail with probability $1-p$. The internal actions $\tau_t$ and $\tau_h$ are enabled according to the result of $\flip_p$. Hence only one of $\tau_h$ and $\tau_t$ will be enabled just after the coin flip. Since $\tau_t$ and $\tau_h$ are internal actions, the choice is internal and based upon the outcome of $\flip_p$. The important facts here are that the choice after $\flip_p$ is internal so could be based on the probabilistic outcome of $\flip_p$ and that the environment cannot interfere with that choice. These two behavioural characteristics are what we consider to be the most general features of probability in a concurrent setting and they are those which we axiomatise and record in our concrete model.


Next, we impose some conditions on the automata to ensure soundness. 

\begin{itemize}
\item[-]\label{hc:reachable} reachability: every state of the automaton is reachable by following a finite path from the initial state. 

\item[-]\label{hc:initial} initiality: there is no transition that leads to the initial state. This means that $a^*$ corresponds to the automata associated to $1 + aa^*$ rather than a self loop labeled by $a\in\Sigma$.
\end{itemize}

We denote by $\aut$ the set of automata satisfying these two conditions. The next step is to define the operators that act on $\aut$. We use the standard inductive construction found in~\cite{Coh09,Gla90,Rab11} and the diagrams illustrating the constructions are given in the appendix.
\begin{itemize}
\item[]\textbf{Deadlock: $0$}\\ This is the automaton that has only one state, namely the initial state, and no transition at all. It is the tuple $(\{i\}, \emptyset,i,\emptyset)$.
\item[]\textbf{Skip: $1$}\\ This is the automaton that has only one state $i$ which is both initial and final. This automaton has no transition i.e. is denoted by $(\{i\},\emptyset,i,\{i\})$.
\item[]\textbf{Single action:} \\
The automata associated with $a$ is $i\trans{a} \circ$ where $i$ is the initial state and $\circ$ is a final state. It is the tuple $(\{i,\circ\},\{i\trans{a} \circ\}, i, \{\circ\})$.
\item[]\textbf{Addition: $P+Q$}\\ This automaton is obtained using the standard construction of identifying the initial states of $P$ and $Q$. (This is possible due to the initiality property.) 
\item[]\textbf{Multiplication: $PQ$ (or $P\cdot Q$)}\\ This automaton is constructed in the standard way of identifying copies of the initial state of $Q$ with final states of $P$.
%
\item[]\textbf{Concurrency: $P\pr{A} Q$} \\ This automaton is constructed as in CSP~\cite{Hoa78}. It is a sub-automaton of the Cartesian product of $P$ and $Q$. The initial state is $(i_P,i_Q)$ and final states are reachable elements of $F_P\times F_Q$. Notice that the set $A$ never contains probabilistic actions. Further explanation about $\pr{A} $ is given below.
%

\item[]\textbf{Kleene star: $P^*$} \\This automaton is the result of repeating $P$ allowing a successful termination after each  (possibly empty) full execution of $P$. The initial state of $P^*$ is final and copies of the initial state of $P$ are identified with the final states of $P$.

\end{itemize}

All automata begin with an initial state and end in some final or deadlock state. Our main use of final states is in the construction of sequential composition and Kleene star.

The concurrency operator $\pr{A} $ synchronises transitions labeled by an action in $A$ and interleaves the others (including internal transitions). As in CSP, a synchronised transition waits for a corresponding synchronisation action from the other argument of $\pr{A} $. This is another reason we do not have $1\pr{\{a\}} P = P$ because if $P = i_P\trans{a} \circ$ and $i_P$ is not a final state, then $$1\pr{\{a\}} P = (\{(i,i_P)\}, \emptyset, (i,i_P),\emptyset) = 0.$$


\begin{proposition}\label{pro:hc-welldef}
These operations of weak concurrent Kleene algebra are well defined on $\aut$ that is if $P,Q\in\aut$ then $P+Q, PQ, P\pr{A} Q$ and $P^*$ are elements of $\aut$.
\end{proposition}

The proof consists of checking that $P+Q, PQ, P\|Q$ and $P^*$ satisfy the reachability and initiality conditions whenever $P$ and $Q$ satisfy the same conditions. (See Proposition \ref{apro:stability} in the appendix).

In the sequel, whenever we use an unframed concurrency operator $\|$, we mean that the frame $A$ has been given and remains fixed.

\subsection{Equivalence}\label{subsec:notion-of-equality}

The previous subsection has given us the objects and operators needed to construct our concrete model. Next we turn to the interpretation of equality for our concrete interpretation.

Following the works found in~\cite{Coh09,Rab11,Mil71}, we again use a simulation-like relation to define valid equations in the concrete model. More precisely, due to the presence of internal actions, we will use an \textit{$\eta$-simulation} as the basis for our equivalence. 

Before we give the definition of simulation, we need the following notation. Given the state $x$ and $y$, we write $x\Rightarrow y$ if there exists a path, possibly empty, from $x$ to $y$ such that it is labelled by internal actions only. This notation is also used in~\cite{Gla90} with the same meaning.
\begin{definition}\label{df:sim}
Let $P,Q$ be automata, a relation $S\subseteq P\times Q$ (or $S:P\rightarrow Q)$ is called \textbf{$\eta$-simulation} if 
\begin{itemize}
\item[--] $(i_P,i_Q)\in S$,
\item[--] if $(x,y)\in S$ and $x\trans{a} x'$ then 
\begin{itemize}
\item[a)] if $a$ is internal then there exits $y'$ such that $y\Rightarrow y'$ and $(x',y')\in S$,
\item[b)] if $a$ is external then there exists $y_1$ and $y'$ in $Q$ such that $y\Rightarrow y_1\trans{a} y'$ and $(x,y_1)\in S$ and $(x',y')\in S$.
\end{itemize} 
\item[--] if $(x,y)\in S$ and $x\in F_P$ then $y\in F_Q$.
\end{itemize}
A simulation $S$ is \textbf{rooted} if $(i_P,y)\in S$ implies $y = i_Q$. If there is a rooted simulation from $P$ to $Q$ then we say that $P$ is simulated by $Q$ and we write $P\leq Q$. Two processes $P$ and $Q$ are \textbf{simulation equivalent} if $P\leq Q$ and $Q\leq P$, and we write $P\equiv Q$. In the sequel, rooted any $\eta$-simulation will be referred simply as a simulation.
\end{definition}


Relations satisfying Definition \ref{df:sim} are also $\eta$-simulation in the sense of~\cite{Gla90} where property (a) is replaced by:
\begin{equation}\label{pr:property-a}
\textrm{if } a \textrm{ is internal then } (x',y)\in S.
\end{equation}
The identity relation (drawn as dotted arrow) in the following diagram is a simulation relation satisfying Definition \ref{df:sim}, but it is not a simulation in the sense of~\cite{Gla90}.
\begin{figure}[h]
$$\xymatrix{
 \ar[d]^\tau\ars[rr]&&\ar[d]^\tau\\
 \circ\ars[rr]&&\circ
}$$
\end{figure}
We need the identity relation to be a simulation here because in our proof of soundness, more complex simulations are constructed from identity relations.

\begin{proposition}\label{pro:eta-sim-welldef}
The following statements hold.
\begin{enumerate}
\item The relational composition of two rooted $\eta$-simulations is again a rooted $\eta$-simulation. That is, if $S,T$ are rooted $\eta$-simulations then $S\circ T$ is also a rooted $\eta$-simulation, where $\circ$ denotes relational composition.
\item The simulation relation $\leq$ is a preorder on $\aut$.
\end{enumerate}
\end{proposition}

Proposition \ref{pro:eta-sim-welldef} is proven in Proposition \ref{apro:sim-equivalence} of the appendix.

Therefore, $\equiv$ as determined by Definition \ref{df:sim} is an equivalence. In fact, we prove in the following proposition that it is a congruence with respect to $+$. 

\begin{proposition}\label{pro:sim-cong}
The equivalence relation $\equiv$ is a congruence with respect to $+$ and $P\leq Q$ iff $P + Q\equiv Q$.
\end{proposition}

The proof adapts and extends the one found in \cite{Gla90} and the specialised version for our case is Proposition \ref{apro:sim-congruence} in the appendix.

It is well documented that $\eta$-simulation is not a congruence without the rootedness condition~\cite{Gla90}. A typical example is given by the expressions $\tau a + \tau b$ and $\tau(a +b)$. The automata associated to these expressions are equivalent under non-rooted $\eta$-simulation.


The manipulation of probabilistic actions is also an important facet of our model. We assume that probabilistic actions are not synchronised and in that respect they are similar to internal actions. However probabilistic actions cannot be treated as internal as the following examples illustrates. Consider the action $\flip_{1/2}$ which flips a fair coin. If $\flip$ is an internal action then the inequality 
$$(\flip_{1/2})(\tau a + \tau b)\leq (\flip_{1/2})\tau a + (\flip_{1/2})\tau b$$
would be valid when interpreted in the concrete model. In other words, we would have the following simulation:
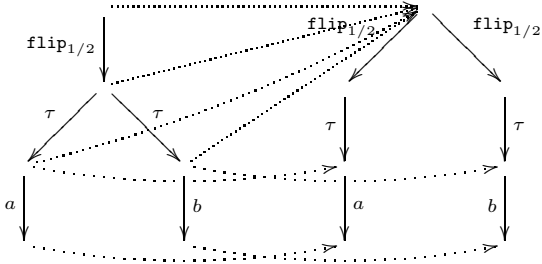
\begin{figure}[h]
$$
\xymatrix{
&\ar[d]_{\flip_{1/2}}\ars[rrrr]&&& 				&\ar[dl]_{\flip_{1/2}}\ar[dr]^{\flip_{1/2}}&\\
&\ar[ld]_{\tau}\ars[rrurr]\ar[dr]^{\tau}&&&		\ar[d]_{\tau}&&\ar[d]^{\tau}\\
\ar[d]_a\ars@/_/[rrrr]\ars@/_/[uurrrrr]&&\ar[d]^b\ars[uurrr]\ars@/_/[rrrr]&&	\ar[d]^a&&\ar[d]_b\\
\ars@/_/[rrrr]&&\ars@/_/[rrrr]&&										&&
}
$$
\end{figure}

But this relationship (which implies distributivity of $\flip_p$ through $+$) does not respect the desired behaviour of probability which, as we explained earlier, satisfies only a weaker form of distributivity.
Whence, we assume that probabilistic actions such as $\flip_{1/2}$ are among the external actions which will never be synchronised.

\section{Soundness}\label{sec:soundness}

In this section, we prove that the set $\aut$ endowed with the operators defined in Subsection \ref{subsec:semantic-space} modulo rooted $\eta$-simulation equivalence (Subsection \ref{subsec:notion-of-equality}) forms a weak concurrent Kleene algebra. 

The first part is to prove that $\aut$ is a weak probabilistic Kleene algebra.

\begin{proposition}\label{pro:pka-soundness}
$(\aut,+,\cdot,*,0,1)$ is a weak probabilistic Kleene algebra.
\end{proposition}

The proof consists of detailed verifications of the axioms for weak probabilistic Kleene algebra (see Proposition \ref{apro:weak-pka} in the appendix).

The second part consists of proving that $\|$ satisfies the equations (\ref{eq:par-assoc}-\ref{eq:exchange-law}). Associativity depends heavily on the fact that both concurrent compositions involved in $x\|y\|z$ have the same frame set. For instance, let $\Sigma = \{a,b,c\}$. The identities
$$(a\pr{\{a\}} b ) \pr{\{c\}} a = ab0 + ba0$$
and
$$a\pr{\{a\}} (b  \pr{\{c\}} a) = ab + ba$$
are valid in the concrete model. Hence, the first process will always go into a deadlock state though the second one will always terminate successfully. Therefore, to have associativity, the concurrency operator must have a fixed frame.

\begin{proposition}\label{pro:soundness-par}
$(\aut, +,\cdot,\pr{A} ,1 )$ satisfies equations (\ref{eq:par-assoc}- \ref{eq:exchange-law}) modulo rooted $\eta$-simulation equivalence for any set of synchronisable actions $A\subseteq\Sigma$ (i.e. no probabilistic actions).
\end{proposition}

Associativity is mainly a consequence of the fact that there is only one frame for $\|$. The other axioms need to be checked thoroughly (see Proposition \ref{apro:parallel-algebra}).

Our soundness result directly follows from these two propositions.

\begin{theorem}
$(\aut,+,\cdot,\pr{A} ,*,0,1)$ is a weak concurrent Kleene algebra for any set of synchronisable actions $A\subseteq\Sigma$.
\end{theorem}

In this theorem, the frame $A$ is fixed beforehand. In other words, a model of weak concurrent Kleene algebra is constructed for each possible choice of $A$. In particular, if $A$ is empty then the concurrency operator is interleaving all actions i.e. no actions are synchronised. This particular model satisfies the identity $1\pr{\emptyset} x = x$ of the original concurrent Kleene algebra found in~\cite{Hoa09}. 

The sequential and concurrent composition actually have stronger properties in the concrete model. If we consider finite automata only --- automata with finitely many states and transitions--- then we show that these two operators are \textit{conditionally Scott continuous} in the sense of~\cite{Rab11} (see Proposition \ref{pro:mult-cont} and \ref{pro:par-continuous} in the appendix).

\section{Relationship to Probabilistic Processes}\label{sec:probabilistic-aut}

Firstly, it is shown in~\cite{Mci04} that a  probabilistic choice $a\pc{p} b$ simulates the nondeterministic choice $a+b$. A similar result also holds in our setting. In the absence of internal transitions, simulation has been also defined elsewhere~\cite{Coh09,Gla90,Rab11} which we will refer to as strong simulation. Recall that $(\flip_p)a+ (\flip_p)b\leq(\flip_p)(a+b)$ is a general property of probabilistic Kleene algebra~\cite{Mci05} so it is valid under strong simulation equivalence~\cite{Coh09,Rab11}. Due to the absence of internal actions,  the middle part of the diagram of Figure \ref{fig:figure1} does not exist with respect to strong simulation equivalence.

In the context of Definition \ref{df:sim}, the right-hand simulation of Figure \ref{fig:figure1} 
is the refinement of probabilistic choice by nondeterminism. This example gives an explicit distinction between $(\flip_p)(a+b)$ and $(\flip_p) a+ (\flip_p) b$ by considering the fact that the choice in $(\flip_p) (a + b)$ can depend on the probabilistic outcome of $(\flip_p)$, but this is not the case for $(\flip_p) a + (\flip_p)b$.
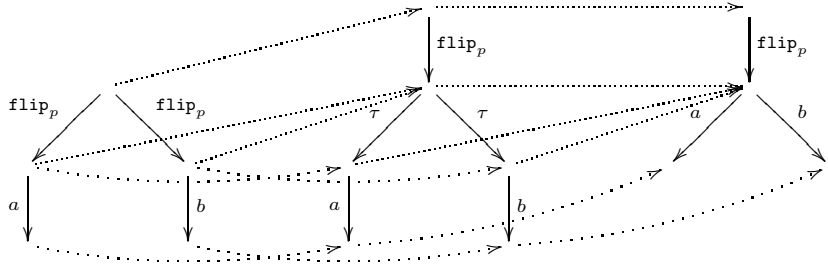
\begin{figure*}
$$
\xymatrix{
&  &	&&& \ar[d]^{\flip_p}\ars[rrrr] & && &\ar[d]^{\flip_p}&\\
	&\ar[dl]_{\flip_p}\ar[dr]^{\flip_p}\ars[urrrr]&	&&& \ar[dl]_{\tau}\ar[dr]^{\tau}\ars[rrrr] & && &\ar[ld]_{a}\ar[rd]^b &\\
\ar[d]_a\ars[urrrrr]\ars@/_/[rrrr]&&\ar[d]^b\ars@/_/[rrrr]\ars[urrr]	&&\ar[d]_{a}\ars[urrrrr]&&\ar[d]^{b}\ars[urrr] && & &\\
\ars@/_/[rrrr]	&&\ars@/_/[rrrr] 	&& \ars@/_/[urrrr]& &\ars@/_/[urrrr] && & &
}$$
\caption{Refinements between probabilistic choice and nondeterminism.}\label{fig:figure1}
\end{figure*}

Secondly, we discuss about the relationship between our concrete model and probabilisitic automata. Remind that our interpretation of probability lies in the use of actions that implicitly contain probabilistic information. In its most general form, a probabilistic choice between $n$ possibilities can be written as 

$$\flip_{p_1,\dots,p_n}\cdot(\tau_1\cdot a_1+\dots+\tau_n\cdot a_n)$$
where $\sum_ip_i = 1$. In this algebraic expression, we implicitly ensure that each guard $\tau_i$ is enabled with a corresponding probability $p_i$. Therefore, if these $\tau_i$'s are not found directly after the execution of the probabilistic action then matching them with the corresponding $p_i$ becomes a difficult task.  We call $p$-automaton~\footnote{The name $p$-automata describes probabilistic automata and as we will see later on, there is a relationship between the two of them.} a transition system as per the definition of Subsection \ref{subsec:semantic-space} such that if a probabilistic action has associated $\tau$ transitions then all of them follow that action directly. 

Another complication also arises from the use of these $\tau_i$'s. Consider the following two processes $$\flip_{p_1,p_2}\cdot(\tau_1\cdot a+\tau_2\cdot b)$$
and 
$$\flip_{p_1,p_2}\cdot(\tau_1\cdot b + \tau_2\cdot a)$$
where $p_1+p_2 = 1$. We can construct a (bi)simulation relation between the corresponding automata though the probabilities of doing an $a$ are different. Hence we need to modify the definition of $\eta$-simulation (Definition \ref{df:sim}) to account for these particular structure. 

\begin{definition}\label{df:p-sim}
A $p$-simulation $S$ between two $p$-automata $P,Q$ is a $\eta$-simulation  such that if
\begin{itemize}
\item[-] $x\trans{\flip_{p_1,\dots,p_n}} x'\trans{\tau_i} x_i''$ is a transition in $P$, 
\item[-] $y\trans{\flip_{p_1,\dots,p_n}} y'\trans{\tau_i} y_i''$ is a transition in $Q$,
\item[-] and $(x,y)\in S$
\end{itemize}
then $(x_i'',y_i'')\in S$, for each $i=1,\dots,n$.
\end{definition} 

This definition ensures that the probability of doing a certain action from $y$ is greater than doing that action from $x$. With similar proofs as in the previous Sections, we can show that the set of $p$-automata modulo $p$-simulation forms again a weak concurrent Kleene algebra. We denote $p$-$\aut$ the set of $p$-automata modulo $p$-simulation.

We will now show that this definition is a very special case of probabilistic simulation on probabilistic automata. To simplify the comparison, we assume that $\tau$ transitions occur only as part of these probabilistic choices in $p$-automata.


\begin{definition}
A probabilistic automaton is defined as a tuple $(P,\lra,\Delta,F)$ where $P$ is a set of states, $\lra$ is a set of labelled transitions from state to distributions~\footnote{We assume that all distributions are finitely supported.} of states i.e. $\lra\subseteq P\times\Sigma\times\D P$, $\Delta$ is the initial distribution and $F\subseteq P$ is a set of final states. 
\end{definition}

The notion of simulation also exists for probabilistic automata~\cite{Seg94} and, in particular, simulation and failure simulation is discussed in~\cite{Den07} where they are proven to be equivalent to may and must testing respectively.  

To give a proper definition of probabilistic simulation, we need the following notations which are borrowed from~\cite{Den07} and~\cite{Gla90}. Given a relation $R\subseteq P\times\D Q$, the lifting of $R$ is a relation $\hat{R}\subseteq \D P\times \D Q$ such that $\phi \hat{R} \psi$ iff:
\begin{itemize}
\item[-] $\phi = \sum_xp_x\delta_x$,~\footnote{We denote by $\delta_x$ the point distribution concentrated on $x$.}
\item[-] for each $x\in\supp(\phi)$ (the support of $\phi$) there exists $\psi_x\in\D Q$ such that $x R\psi_x$,
\item[-] $\psi = \sum_xp_x\psi_x$.
\end{itemize}
Similarly, the lifting of a transition relation $\trans{\tau} $ is denote $\trans{\hat{\tau}} $ whose reflexive transitive closure is denote $\ttrans{\hat{\tau}} $. For each external action $a$, we write $\ttrans{\hat{a}} $ for the sequence $\ttrans{\hat{\tau}} \trans{a} $.
\begin{definition}\label{df:probsim}
A probabilistic simulation $S$ between two probabilistic automata $P$ and $Q$ is a relation $S\subseteq R\times\D Q$ such that:
\begin{itemize}
\item[-] $(\Delta_P,\Delta_Q)\in \hat{S}$,
\item[-] if $(x,\psi)\in S$ and $x\trans{a} \phi$ then there exists $\psi'\in\D Q$ such that $\psi\ttrans{\hat{a}} \psi'$ and $(\phi,\psi')\in\hat{S}$ (for every $a\in\Sigma\cup\{\tau\}$).
\item[-] if $x\in F_P$ and $(x,\psi)\in S$ then $\supp(\psi)\subseteq F_Q$.
\end{itemize} 
\end{definition}

we denote by $\paut$ the set of probabilistic automata modulo simulation equivalence.

We can now construct a mapping $\epsilon:\praut\rightarrow \paut$ such that each instance of structure similar to $\flip_{p_1,\dots, p_n}\cdot(\tau_1\cdot a_1 + \dots + \tau_n\cdot a_n)$ is collapsed into probabilistic transitions. More precisely, let $P\in\praut$ and $\lra$ be its transition relation. The automaton $\epsilon(P)$ has the same state space as $P$ (up to accessibility with respect to the transitions of $\epsilon(P)$). The initial distribution of $\epsilon(P)$ is $\delta_{i_P}$ and the set of final states of $\epsilon(P)$ is $F_P$ again~\footnote{Notice that by  assuming the structure $\flip_{p_1,\dots,p_n}\cdot(\tau_1\cdot a_1 + \dots + \tau_n\cdot a_n$, the state between the flip action the corresponding $\tau$ transitions is never a final state. Hence we are safe to use $F_P$ as the final state of $\epsilon(P)$} .

The set of transitions $\lra_{\epsilon(P)}$ is constructed as follow. Let $x\trans{a} x'$ be a transition of $P$, there are two possible cases:
\begin{itemize}
\item[a)] if $a$ is probabilistic i.e. of the form $\flip_{p_1,\dots,p_n}$ and is followed by the $\tau_i$'s, then  the transition
$$x\trans{\tau} p_1\delta_{x_1'}+\dots+p_n\delta_{x_n'}$$
is in $\lra_{\epsilon(P)}$ where $x'\trans{\tau_i} x_i'$ is a transition in $P$.
\item[b)] else the transition $x\trans{a} x'$ is in $\lra_{\epsilon(P)}$.
\end{itemize}

We now prove that $\epsilon$ is a monotonic function from $\praut$ to $\paut$.

\begin{proposition}\label{pro:cor}
If $P\leq Q$ then $\epsilon(P)\leq\epsilon(Q)$.
\end{proposition}

\begin{proof}
Assume that $S$ is a $p$-simulation from $P$ to $Q$. Consider the exact same relation but restricted to the state space of $\epsilon(P)$ and $\epsilon(Q)$. We show that this restriction is a probabilistic simulation.
\begin{itemize}
\item[-] Obviously, $(\delta_{i_P},\delta_{i_Q})\in \hat{S}$.
\item[-] Let $x\trans{a} \phi$ and $(x,\psi)\in\hat{S}$. Since $\tau$ transitions only occur as part of probabilistic choices, we have two possibilities:
\begin{itemize}
\item $x\trans{\tau} p_1\delta_{x_1'}+\dots +p_n\delta_{x_n'}$ is a transition of $\epsilon(P)$ and $(x,\psi)\in S$ where $\psi = \delta_y$. Since $(x,y)$ belongs to the original $S$. In this case, $y\trans{\tau} p_1\delta_{y_1'}+\dots +p_n\delta_{y_n'}$ is a transition of $\epsilon(Q)$ and each $(x_i',y_i')$ belongs to the original $S$ (Definition of $p$-simulation).
\item $x\trans{a} x'$ and $a$ is an external action. Therefore there are two possibilities again, $y\trans{\tau_i} y_i\trans{a} y'$ or $y\trans{a} y'$.  In both cases, we have $(x',y')\in S$.
\end{itemize}
\item[-] Conservation of final states follows easily from the fact that $S$ is a $p$-simulation.\qedhere
\end{itemize}
\end{proof}

Since our Definition (\ref{df:probsim}) implies the definition of probabilistic simulation in~\cite{Den07}, we conclude that maximal probability of doing a particular action in $p$-automata is increased by $p$-simulation. This remark provides a formal justification of our earlier example. That is, Equation~(\ref{eq:vm}) ensures that the maximal probability that a buyer will be satisfied when using the probabilistic vending machine is at least $1/2$ because the maximal probability of a trace containing $\tea$ in the automata described by 
$$\coin\cdot\flip\cdot(\tau_h\cdot(\tea+1) + \tau_t$$
is $1/2$. 

In the proof of proposition \ref{pro:cor}, the simulation constructed is a very particular case of probabilistic simulation so it is too weak to establish certain relationships between $p$-automata. For instance, the automaton represented by $a\pc{p} (a\pc{q} b)$ should be equivalent to $a\pc{p+q-pq} b$ but Definition \ref{df:p-sim} will not provide such equality. This line of research is part of our future work where we will study proper probabilistic automata and simulations against weak concurrent Kleene algebra.

\section{Algebraic Testing}\label{sec:algebraic-testing}

In this section, we describe an algebraic treatment of \emph{testing}. Testing is a natural ordering for processes that was studied first in~\cite{Nic83}. The idea is to ``measure" the behaviour of the process with respect to the environment. In other words, given two processes $x$ and $y$ and a set of test processes $T$, the goal is to compare the processes $x\|t$ and $y\|t$ for every $t\in T$. In our case, the set $T$ will contain all processes.

We consider a function $o$ from the set of terms to the set of internal expressions $I = \{x\ |\ x\leq 1\}$. The function $o:T_\Sigma\rightarrow I$ is defined by
$$\begin{array}{lll}
o(x) = x\textrm{ if }x\in I & &o(st) = o(s)o(t)\\
o(a) = \tau\textrm{ for any a }\in\Sigma-I && o(s^*) = 1\\
o(s+t) = o(s)+ o(t) &&o(s\|t)\leq o(s)o(t)
\end{array}$$

In the model, the function $o$ is interpreted by substituting each external action with the internal action $\tau$ ($o(a) = \tau$ for any $a\in\Sigma-I$). Then any final state is labelled by $1$ and deadlock states are labelled by $0$. Inductively, we label a state that leads to some final state by $1$, else it is labelled by $0$. This is motivated by the fact that $x0=0$ for any $x\in I$ so each transition leading to \textit{deadlock states only} will be removed. Therefore, only states labelled by $1$ will remain and the transitions between them. Hence, $o(s)\neq 0$ iff the resulting automaton contains at least one state labelled by $1$. In other words, $o(s) = 0$ iff $x$ \textit{must not terminate successfully}. 

Without loss of generality (by considering automata modulo simulation), we assume that $\tau$ is the only internal action in $\Sigma$ and it satisfies $\tau\tau = \tau$. This equation is valid in the concrete model.

The existence of a well-defined function $o$ satisfying these conditions depends on our definition of simulation. That is, we can show that if $P\leq Q$ then $o(P)\leq o(Q)$ where we have abused notation by writing $o(P)$ as the application of $o$ on the term associated to $P$. A detailed discussion about this can be found in the appendix under Remark \ref{rem:remark-o}.

\begin{definition}
The \textit{may testing order} is given by
$$x\may{} y\quad \textrm{ iff }\quad \forall t\in T_\Sigma.\left[o(y\| t) = 0\Rightarrow o(x\| t) = 0\right].~\footnote{Notice $\|$ should be framed because some external actions are not synchronised. But in the setting of testing, we can also assume that all external actions are synchronised which permits to follow up all external actions present in the process.}$$
\end{definition}

We now provide some results about algebraic may testing. It follows from monotonicity of $\|$ with respect to $\leq$ (Proposition \ref{pro:elementary-consequences}) that may ordering $\may{} $ is weaker than the rooted $\eta$-simulation order.

\begin{proposition}
$x\leq y$ implies $x\may{} y$.
\end{proposition}

In fact, $\may{} $ is too weak compared to $\leq$: may testing is equivalent to language equivalence. Given a term $s$, the language $Tr(s)$ of $s$ is the set of finite words formed by external actions and are accepted by the automata represented by $s$. In other word, it is the set of finite traces in the sense of CSP which lead to final states. The precise definition of this language equivalence can be found in the appendix and so is the proof of the following proposition (Proposition \ref{apro:may-language} of the appendix).
\begin{proposition}\label{pro:may-equals-language}
In $\aut$, $\may{} $ reduces to language equivalence.
\end{proposition}

We have shown that $\may{} $ is equivalent to language equivalence and hence it is weaker than our simulation order. This is also a consequence of the fact that our study of may testing is done in a qualitative way because the probabilities are found implicitly within actions. A quantitative study of probabilistic testing orders can be found in~\cite{Den07}.

\section{Case Study: Rabin's Choice Coordination}\label{sec:rabin-protocol}

The problem of choice coordination is well known in the area of distributed systems. It usually appears in the form of processes voting for a common goal among some possibilities. Rabin has proposed a probabilistic protocol which solves the problem~\cite{Rab82} and a sequential specification can be found in~\cite{Mci04}.

We specify the protocol in our algebra and prove that a fully concurrent specification is equivalent to a sequential one. Once this has been done, the full verification can proceed by reusing the techniques for sequential reasoning~\cite{Mci04}.

The protocol consists of a set of tourists and two places: a church $C$ and a museum $M$. Each tourist has a notepad where he keeps track of an integer $k$. Each place has a board where tourists can read and write. We denote by $L$ (resp. $R$) the value on the church board (resp. museum board). 

In this section, we use $\cdot$ again for the sequential composition to make the specifications clearer.

\begin{itemize}
\item The church is specified as $C = (c!L)^*\cdot(c?L)$ where the channel $c$ represents the church's door. $c!L$ means that the value of $L$ is available to be read in the channel $c$ and $c?L$ waits for an input which is used as value for $L$ in the subsequent process.

In other words, each tourist can read as many times as they want from the church board but write on it only once. Repeated writing will be considered in the specification of the protocol.

Similarly, the museum is specified as $M = (m!R)^*\cdot(m?R)$.

\item Each tourist is specified as $P(\alpha,k)$ where $\alpha\in\{c,m\}$ is the door before which the tourist currently stands and $k$ is the actual value written on his notepad. A detailed description of $P$ can be found in the appendix but roughly, we have 
$$P(\alpha,k) = (\alpha?K)\cdot\mathtt{rabin}\cdot[\alpha:=\underline{\alpha}]~\footnote{Any action written within square brackets will denote internal action (see appendix for the detailed specification).}$$
where $\underline c = m$ and $\underline m = c$. In other words, the tourist reads the value on the place specified by $\alpha$, executes Rabin's protocol \texttt{rabin} and then goes to the other place. Notice that the process $\texttt{rabin}$ contains the probabilistic component of Rabin's protocol. Essentially, it describes the rules that are used by each tourist to update their actual value for $k$ with respect to the value on the board and vice versa. 

The whole specification of the protocol executed by each tourist is described by the automata of Figure \ref{fig:rabin}
\begin{center}
\begin{figure*}
$$\xymatrix{
&P(\alpha,k)\ar[d]_{\alpha?K}&&\\
&\ar[dl]_{[K=here]}\ar[dr]^{[K\neq here]}&&\\
\bullet	& \ar[l]^{\alpha!here}& \ar[l]_{[k>K]}\ar[dl]^{[k<K]}\ar[d]^{[k=K]}&\\
	& \ar[dl]^{[k:=K]} 			&\ar[d]^{\flip_{1/2}}& \\
\ar[d]_{\alpha!k}& &\ar[dl]_{\tau_h}\ar[dr]^{\tau_t} &\\
\ar[d]_{[\alpha := \underline\alpha]}& \ar[dr]_{[k := K+2]}& &\ar[dl]^{[k:=\overline{K+2}]}\\
\circ&&\ar[d]^{\alpha!k}&\\
& &\ar[d]^{[\alpha := \underline{\alpha}]} &\\
& &\circ &
}$$
\caption{$p$-automaton that describes the protocol $P(\alpha,k)$ executed by each tourist.}\label{fig:rabin}
\end{figure*}
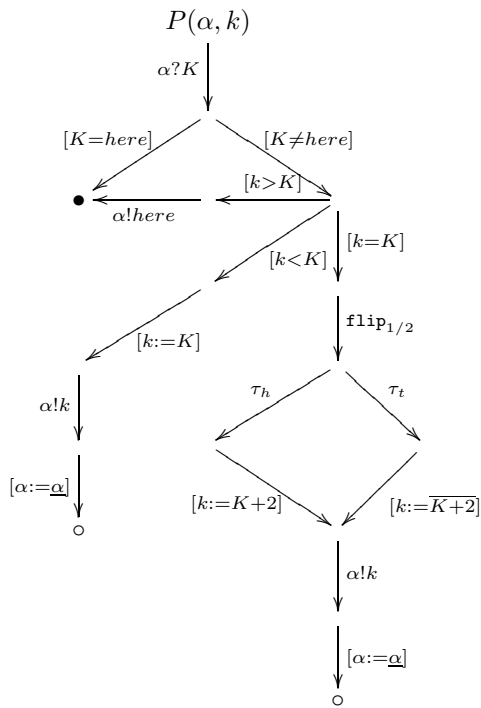
\end{center}

\end{itemize}

We are ready to specify the whole system. Assume we have two tourists $P$ and $Q$ (our result generalises easily to $n$ tourists). The tourists' joint action is specified as $(P + Q)^*$. This ensures that when a tourist has started his turn by reading the board, he will not be interrupted by any other tourist until he is done and goes inside the current place or to the other place. This condition is crucial for the protocol to work properly. 

The actions of the locations process are specified by $(M+C)^*$ which ensures that each tourist can be at one place at a time only --- this is a physical constraint. Now, the whole system is specified by 
\begin{equation}\label{eq:spec}
\mathtt{init}\cdot\left([P(\alpha,u) + Q(\beta,v)]^*\pr{\{c,m\} } (M + C)^*\right)
\end{equation}

where $\mathtt{init}$ is the initialisation of the values on the boards, notepads and initial locations. Specification \ref{eq:spec} describes the most arbitrary behaviour of the tourists compatible with visiting and interacting with the locations in the manner described above. Rabin's design of the protocol means that this behaviour is equivalent to a serialised execution where first one location is visited, followed by the other.   We can write that behaviour behaviour as $[((P+Q)\|M)^*((P+Q)\|C)^*]^*$, where (for this section only)  we denote the concurrency operator by $\|$ instead of $\pr{\{c,m\}} $ to make the notation lighter. The next theorem says that this more uniform execution is included in $S=  [P(\alpha,u) + Q(\beta,v)]^*\| (M + C)^*$, described by Specification \ref{eq:spec}.


\begin{theorem}\label{pro:dupl}
We have $$S \geq [((P+Q)\|M)^*((P+Q)\|C)^*]^*$$
\end{theorem}

The proof is a simple application of Proposition \ref{pro:elementary-consequences}. Theorem \ref{pro:dupl} means $S$ could execute all possible actions related to door $M$, and then those at door $C$, and then back to door $M$ and so one. In fact, we can also prove the converse i.e. Proposition \ref{pro:dupl} could be strengthen to equality. But for that, we need the continuity of the operators $\cdot$ and $\|$.

\begin{theorem}\label{thm:rabin}
In the concrete model, the specification of Rabin's protocol satisfies $$S = [((P+Q)\|M)^*((P+Q)\|C)^*]^*$$
\end{theorem}
The proof of this theorem depends heavily on the fact that the concurrent and sequential compositions are continuous in the the concrete model. The complete proof can be found in the appendix.

In the proof, if we stopped at the distribution over $\|$, we obtain the equivalent specification
$$S = [(P+Q)\|M + (P+Q)\|C]^*$$
which describes a simpler  situation where P or Q interacts at the Museum or at the Church. This is similar to the sequential version found in~\cite{Mci04}, which can be treated by standard probabilistic invariants to complete a full probabilistic analysis of the protocol. 

\section{Conclusion}

An algebraic account of probabilistic and concurrent system has been presented in this paper. The idea was to combine probabilistic and concurrent Kleene algebra. A soundness result with respect to automata and rooted $\eta$-simulation has been provided. The concrete model ensures not only the consistency of the axioms but provides also a semantic space for systems exhibiting probabilistic, nondeterministic and concurrent behaviour. We also showed that the model has stronger properties than just the algebraic axiomatisation. For instance, sequential and concurrent compositions are both continuous in the case of finite automata.

We provided some applications of the framework. An algebraic account of may testing has been discussed in Section \ref{sec:algebraic-testing}. It was shown that may ordering reduces to language equivalence. 

We also provided a case study of Rabin's solution to the choice coordination problem. A concurrent specification was provided and it was shown to be structurally equivalent to the sequential one given in~\cite{Mci04}. 

Though the algebra was proven to be powerful enough to derive non-trivial properties for concrete protocols, the concrete model still needs to be refined. For instance, the inclusion of tests is important especially for the construction of probabilistic choices. Tests need to be introduced carefully because their algebraic characterisation are subtle due to presence of probability. We also need to improve and refine the manipulation of quantitative properties in the model as part of our future work. 

Finally, it is customary to motivate automated support for algebraic approaches. The axioms system for weak concurrent Kleene algebra is entirely first-order, therefore proof automation is supported and automatised version of our algebraic proofs can be found in our repository.

\bibliographystyle{plain}
\bibliography{references-comp}

\begin{thebibliography}{10}

\bibitem{Coh09}
E.~Cohen.
\newblock Weak {K}leene algebra is sound and (possibly) complete for
  simulation.
\newblock {\em CoRR}, abs/0910.1028, 2009.

\bibitem{Con71}
J.~H. Conway.
\newblock {\em Regular Algebra and Finite Machines}.
\newblock Chapman and Hall, Mathematics series, 1971.

\bibitem{Den07}
Y.~Deng and R.~Van~Glabbeek.
\newblock Characterising testing preorders for finite probabilistic processes.
\newblock In {\em In LICS’07: Proceedings of the 22nd Annual IEEE Symposium
  on Logic in Computer Science. IEEE Computer Society Press, Los Alamitos, CA},
  pages 313--325, 2007.

\bibitem{Hoa78}
C.~A.~R. Hoare.
\newblock Communicating sequential processes.
\newblock {\em Commun. ACM}, 21:666--677, August 1978.

\bibitem{Hoa09}
C.~A.~R. Hoare, B.~M\"{o}ller, and I.~Struth, G.and~Wehrman.
\newblock Concurrent {K}leene algebra.
\newblock In {\em Proceedings of the 20th International Conference on
  Concurrency Theory}, CONCUR 2009, pages 399--414, Berlin, Heidelberg, 2009.
  Springer-Verlag.

\bibitem{Koz94}
D.~Kozen.
\newblock A completeness theorem for {K}leene algebras and the algebra of
  regular events.
\newblock {\em Infor. and Comput.}, 110(2):366--390, May 1994.

\bibitem{Koz00a}
D.~Kozen.
\newblock On {H}oare logic and {K}leene algebra with tests.
\newblock {\em Trans. Computational Logic}, 1(1):60--76, July 2000.

\bibitem{Koz00b}
D.~Kozen and M.~C. Patron.
\newblock Certification of compiler optimizations using {K}leene algebra with
  tests.
\newblock In John Lloyd, Veronica Dahl, Ulrich Furbach, Manfred Kerber,
  Kung-Kiu Lau, Catuscia Palamidessi, Luis~Moniz Pereira, Yehoshua Sagiv, and
  Peter~J. Stuckey, editors, {\em Proc. 1st Int. Conf. Computational Logic
  (CL2000)}, volume 1861 of {\em LNAI}, pages 568--582, London, July 2000.
  Springer-Verlag.

\bibitem{Rab11}
A.~McIver, T.~M. Rabehaja, and G.~Struth.
\newblock On probabilistic {K}leene algebras, automata and simulations.
\newblock In {\em Proceedings of the 12th international conference on
  Relational and algebraic methods in computer science}, RAMICS'11, pages
  264--279, Berlin, Heidelberg, 2011. Springer-Verlag.

\bibitem{Mci06}
A.~K. McIver, E.~Cohen, and C.~C. Morgan.
\newblock Using probabilistic {K}leene algebra for protocol verification.
\newblock In {\em In Relmics/AKA 2006, volume 4136 of LNCS}. Springer Verlag.

\bibitem{Mci04}
A.~K. McIver and C.~C. Morgan.
\newblock {\em Abstraction, Refinement And Proof For Probabilistic Systems
  (Monographs in Computer Science)}.
\newblock SpringerVerlag, 2004.

\bibitem{Mci05}
A.~K. McIver and T.~Weber.
\newblock Towards automated proof support for probabilistic distributed
  systems.
\newblock In {\em In Proceedings of Logic for Programming and Automated
  Reasoning, volume 3835 of LNAI}, pages 534--548. Springer, 2005.

\bibitem{Mil71}
R.~Milner.
\newblock An algebraic definition of simulation between programs.
\newblock Technical report, Stanford, CA, USA, 1971.

\bibitem{Nic83}
R.~De Nicola and M.~Hennessy.
\newblock Testing equivalence for processes.
\newblock In {\em Proceedings of the 10th Colloquium on Automata, Languages and
  Programming}, pages 548--560, London, UK, 1983. Springer-Verlag.

\bibitem{Rab82}
M.~O. Rabin.
\newblock The choice coordination problem.
\newblock {\em Acta Inf.}, 17:121--134, 1982.

\bibitem{Seg94}
R.~Segala and N.~Lynch.
\newblock Probabilistic simulations for probabilistic processes.
\newblock In {\em Nordic Journal of Computing}, pages 481--496. Springer, 1994.

\bibitem{Gla90}
R.~G. van Glabbeek.
\newblock The linear time-branching time spectrum (extended abstract).
\newblock In J.~C.~M. Baeten and J.~W. Klop, editors, {\em CONCUR 1990}, volume
  458 of {\em LNCS}, pages 278--297. Springer, 1990.

\end{thebibliography}

\newpage
$$$$
\newpage

\appendix

\section*{Appendix}
The following proofs, diagrams, remarks and other results are only included to add further clarification of the contents of the present paper. It is left to the discression of the reviewers to choose whether they will read these proofs or not.

\section{Diagrams, Theorems and Proofs}


%
\textbf{Diagram of the Operators:} The construction are done inductively from $0,1$ and elements of the alphabet $\Sigma$.
\begin{itemize}
\item[-]\textbf{Deadlock:} $0$. \\
This is the automaton that has only one state, no transition and no final state.
\item[-]\textbf{Skip:} $1$ \\
This is the automaton $\circ$ which has only one state which is both initial and final and has no transition.
\item[-]\textbf{Single action:} \\
The automaton associated to $a\in\Sigma$ is $i\trans{a} \circ$ where $i$ is the initial state and $\circ$ is a final state.
\item[-]\textbf{Addition:} $P+Q$.\\
This is constructed by identifying the initial states of $P$ and $Q$. This construction is allowed because of the initiality condition (Figure \ref{fig:p+q}).

\begin{figure*}
$$
\xymatrix{
 	& i_P\ar[dl]_a\ar[d]^b & + & i_Q\ar[d]_c\ar[dr]^d & & = & P'_1& \ar[l]_a\ar[dl]^bi_{P+Q}\ar[r]^d\ar[dr]_c&Q'_2\\
P'_1&P'_2 	&	& Q'_1					& Q'_2& & P'_2& & Q'_1
}
$$
\caption{Automaton for $P+Q$.}\label{fig:p+q}
\end{figure*}
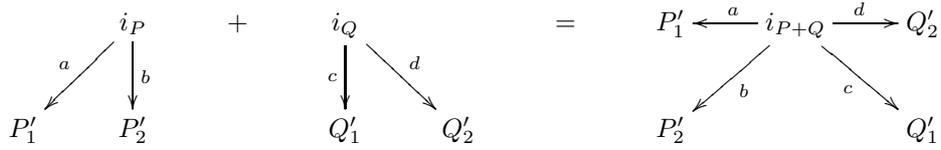

\item[-]\textbf{Multiplication:} $PQ$.\\
This is constructed by identifying each final state of $P$ with the initial state of $Q$ (Figure \ref{fig:pq}). 
\begin{figure*}
$$
\xymatrix{
& P'\ar[d]^a & \cdot & i_Q\ar[d]_c\ar[dr]^d & & = & & P'\ar[d]^a&\\
\dots&\ar[l]_{\dots}\circ&	& Q'_1		& Q'_2& & \dots& \ar[l]_{\dots}\bullet \ar[d]_c\ar[dr]^d&\\
& & & & & & & Q'_1& Q'_2
}
$$
\caption{Automaton for $PQ$. The symbol $\circ$ denotes a final state and the symbol $\bullet$ is final if and only if $i_Q$ is final in $Q$. Notice that this construction is done for each final state of $P$.}
\label{fig:pq}
\end{figure*}
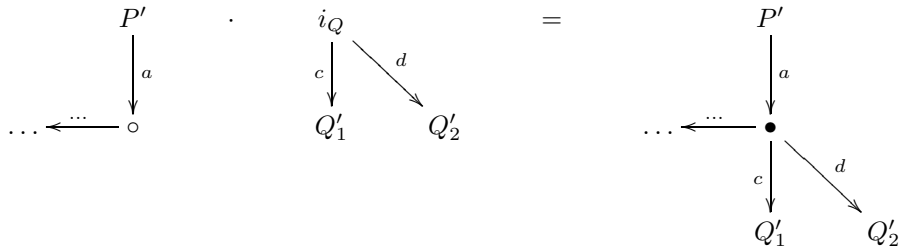

\item[-]\textbf{Concurrency:} $P\pr{A} Q$\\
This is constructed  as a sub-automaton of the Cartesian product of $P$ and $Q$ following CSP~\cite{Hoa78}. Assuming $a\in A$ and $b,d\notin A$, the concurrent composition $P\pr{A } Q$ is inductively constructed as in Figure \ref{fig:p|q}
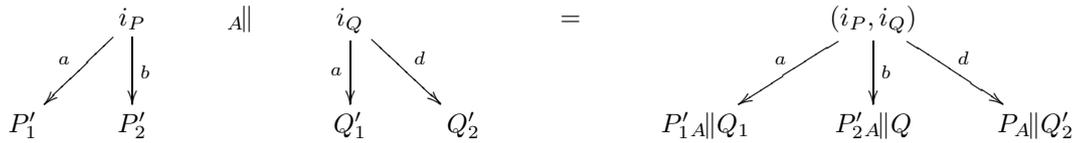
\begin{figure*}
$$
\xymatrix{
 	& i_P\ar[dl]_a\ar[d]^b & \pr{A} & i_Q\ar[d]_a\ar[dr]^d & & =   & & \ar[dl]_a(i_{P},i_Q)\ar[d]^b\ar[dr]^d&\\
P'_1&P'_2 	&	& Q'_1					& Q'_2&  &P'_{1}\ \!\! \pr{A} Q_1& P'_{2}\ \!\!\pr{A}  Q  & P\pr{A } Q'_2
}
$$
\caption{Automaton for $P\pr{A} Q$. The action $a$ has been synchronised and $b,d$ were interleaved. Notice that $b$ or $d$ could be internal. The initial state of the automata is the pair $(i_P,i_Q)$ and the final states are the elements of $F_P\times F_Q$.}
\label{fig:p|q}
\end{figure*}

Notice that $A\subseteq\Sigma$ is a set of synchronised action and does not contain any (strictly) probabilistic actions such as $\flip(p)$, for $p\in]0,1[$.

\item[-]\textbf{Kleene star:} $P^*$\\
This is the result of repeating $P$ allowing a successful termination after each  (possibly empty) full execution of $P$. 
In the diagram of Figure \ref{fig:p*}, we just picture one transition from the initial state and one final state. The construction needs to be performedfor each initial transition and final state.
\begin{figure*}
$$
\xymatrix{
(i_P\ar[r]^a& P'\ar[r]^b& \circ )^* & = &  \circ\ar[r]^a & P'\ar@/^/[r]^b& \circ\ar@/^/[l]^a}
$$
\caption{Automaton for $P^*$.}\label{fig:p*}
\end{figure*}
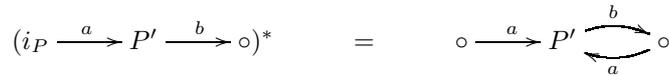
Notice the initial state of $P^*$ is a final state too.
\end{itemize}

\begin{proposition}\label{apro:stability}
These operations are well defined on $\aut$ that is if $P,Q\in\aut$ then $P+Q, PQ, P\pr{A} Q$ and $P^*$ are elements of $\aut$.
\end{proposition}

\begin{proof}
The proof is by induction on the structure of the automata $P$ and $Q$. For the base case, it is obvious that $0,1$ and $i\trans{a} \circ$ satisfy the reachability and initiality conditions.

Let $P,Q\in\aut$. It is easy to see from the diagrams that $P+Q,P\|Q$ and $P^*$ belongs to $\aut$ too. $PQ$ satisfies the initiality condition because the initial state is $i_P$. For reachability, let $x\in Q$. Then $x$ is reachable from $i_Q$ which in turn is reachable from $i_P$ by the definition of sequential composition.
\end{proof}

\begin{proposition}\label{apro:sim-equivalence}
The following statements hold.
\begin{enumerate}
\item The relational composition of two rooted $\eta$-simulations is again a rooted $\eta$-simulation. That is, if $S,T$ are rooted $\eta$-simulations then $S\circ T$ is also a rooted $\eta$-simulation, where $\circ$ denotes relational composition.
\item The simulation relation $\leq$ is a preorder on $\aut$.
\end{enumerate}
\end{proposition}

\begin{proof}
\begin{enumerate}
\item Let $S:P\rightarrow Q$ and $T:Q\rightarrow R$ be simulations and let us show that $ST:P\rightarrow R$ is a simulation.
\begin{itemize}
\item[--] Evidently, $(i,i)\in ST$. 
\item[--] Let $(x,z)\in S T$ and $x\trans{a} x'$. By definition of the relational composition there exists $y\in Q$ such that $(x,y)\in S$ and $(y,z)\in T$.
\begin{itemize}
\item[a)] if $a$ is internal,  there exists $y'\in Q$ such that $y\Rightarrow y'$ and $(x',y')\in S$. Since $y\Rightarrow y'$ consists of a sequence of finite internal transition, there exists $z'\in T$ such that $(y',z')\in T$ and $z\Rightarrow z'$. Hence $(x',z')\in ST$ and $z\Rightarrow z'$.
\item[b)] If $a$ is external, there exists $y_1,y'\in Q$ such that $y\Rightarrow y_1\trans{a} y'$ and $(x,y_1)\in S$ and $(x',y')\in S$. Since $(y,z)\in T$ and $y\Rightarrow y_1$, there exists $z_1\in R$ such that $(y_1,z_1)\in T$ and $z\Rightarrow z_1$. Again, since $T$ is a simulation and $y_1\trans{a} y'$, there exists $z_2,z'\in R$ such that $z_1\Rightarrow z_2\trans{a} z'$ and $(y_1,z_2)\in T$ and $(y',z')\in T$. Hence, by transitivity of $\Rightarrow$, we have $z\Rightarrow z_2\trans{a} z'$ and $(x,z_2)\in ST$ and $(x',z')\in ST$.
\end{itemize}
\item[--] Let $(x,z)\in ST$ and $x\in F_P$, there exists $y\in Q$ such that $(x,y)\in S$ and $(y,z)\in T$. So $y\in F_Q$ and hence $z\in F_R$.
\item[--] Let $(i,z)\in ST$, there exists $y\in Q$ such that $(i,y)\in S$ and $(y,z)\in T$. So $y=i$ and hence $z=i$.
\end{itemize}
\item For reflexivity, the identity relation is a rooted $\eta$-simulation and transitivity follows from 1.
\end{enumerate}
\end{proof}

\begin{proposition}\label{apro:sim-congruence}
The equivalence relation $\equiv$ is a congruence with respect to $+$ and $P\leq Q$ iff $P + Q\equiv Q$.
\end{proposition}

\begin{proof}
Let $S:P\rightarrow Q$ and $S':P'\rightarrow Q'$ be routed $\eta$-simulations. We show that $S\cup S':P+P'\rightarrow Q+Q'$ is again a routed $\eta$-simulation. 
\begin{itemize}
\item[-] Since initial states are identified in the construction of $+$, we have $(i_{P+P'},i_{Q+Q'}) = (i_P,i_Q)\in S\cup S'$.
\item[-] Let $(x,y)\in S\cup S'$ and $x\trans{a} x'$ be a transition of $P$ (the case where this transition belongs to $P'$ is dealt with the exact same way). We have two cases:
\begin{enumerate}
\item if $(x,y)\in S$, then either $a$ is internal and hence $(x',y)\in S$ (so in $S\cup S'$ too) or there exists $y_1,y'\in Q$ such that $y\Rightarrow y_1\trans{a} y'$ is a path in $Q$ and $(x,y_1)\in S$ and $(x',y')\in S$. By definition of $+$, $y\Rightarrow y_1\rightarrow y'$ is again a path in $Q+Q'$ such that $(x,y_1)\in S\cup S'$ and $(x',y')\in S\cup S'$.
\item if $(x,y)\in S'$, then $x = i_P$ because $x\trans{a} x'$ is assumed to be a transition in $P$. Since the initial states are merged, $(i_{P'},y)\in S'$ and therefore $y = i_{Q'} = i_{Q+Q'} = i_Q$. Therefore $(x,y)\in S$ and we are back to Case 1. 
\end{enumerate}
\item[-] Let $(x,y)\in S\cup S'$ and $x\in F_P$ (the case $x\in F_{P'}$ is similar). We have two cases again, $(x,y)\in S$ and $y\in F_{P'}$. Or $(x,y)\in S'$ and then $x\in P\cap P'$. Hence $x=i$ and we are back to the first case again.
\item[-] $S\cup S'$ is rooted because $S$ and $S'$ are both rooted.
\end{itemize}

Now assume $P\leq Q$. Then $P + Q\leq Q + Q\equiv Q$ follows from the fact that $\leq$ is a congruence and the idempotence of $+$ in Proposition \ref{pro:pka-soundness}. Moreover, since $id_Q:Q\rightarrow Q$ is a simulation, we have $P+Q\equiv Q$. Conversely, assume $P + Q\equiv Q$, since $id_P:P\rightarrow P + Q$ is a simulation we have $P\leq Q$ by transitivity of $\leq$.
\end{proof}

\begin{proposition}\label{apro:weak-pka}
$(\aut,+,\cdot,*,0,1)$ is a weak probabilistic Kleene algebra.
\end{proposition}

\begin{proof}
Associativity and commutativity of $+$ and $0+x = x$ follows easily from the fact that $+$ is base on $\cup$. 

\begin{itemize}
\item Idempotence of $+$: since the union is made disjoint, we assume $P_c$ is a copy of $P$ where every states is indexed by $c$. Then $id_P:P\rightarrow P + P_c$ is a simulation and $\{(y,x)\ |\ y = x \textrm{ or } y = x_c\}$ is a simulation from $P+P_c$ to $P$.
\item Associativity of $\cdot$: associativity follows from the same proof found in \cite{Rab11} because identity relations are simulation and our multiplication here is exactly the $\eps$-free version of the multiplication there. 
\item $1$ is neutral for $\cdot$: it follows easily from the construction that $1P = P$ and $P1 = P$.
\item Subdistributivity \ref{eq:subdist}: to show that $PQ + PR\leq P(Q+R)$, it suffices to show that $PR\leq P(Q+R)$ and derive the result from idempotence of $+$. Remind that $id_P$ and $id_Q$ are simulation so it suffices to show that $id_P\cup id_Q:PQ\rightarrow P(Q+R)$ is again a simulation. Obviously, $(i,i)\in S$ and it is rooted and conserves final states. Moreover $\lra_{P(Q+R)}\supseteq  \lra_{PQ}$.
Hence $id_P\cup id_Q$ is a simulation.
\item Right distributivity \ref{eq:rdist}: let $P_c$ be a disjoint copy of $P$, then the relation 
$$S = \{(x,y)\ |\ y = x\textrm{ or } y = x_c\}\nonumber$$
from $(Q+R)P$ to $QP_c + RP$
is rooted and preserves final states. Let $(x,y)\in S$ and $x\trans{a} x'\in\lra_{(Q+R)P}$. Remind that 
\begin{eqnarray}
\lra_{(Q+R)P} = \lra_Q\cup\lra_R\cup\lra_P-\{i\trans{a} z\in\lra_P\} \nonumber\\
\cup\{z\trans{a} z' \ |\ z\in F_Q\cup F_R \textrm{ and } i\trans{a} z'\in\lra_P\nonumber\}
\end{eqnarray}

If the transition belongs to the first three sets then we are done, else we can assume $x\in F_Q$ and $i\trans{a} x'\in\lra_P$ i.e. $y = x$ and $x'\in P$. We have
$$\lra_{QP_c+RP} \supseteq \{z\trans{a} z'_c \ |\ z\in F_Q \textrm{ and } i\trans{a} z'_c\in\lra_{P_c}\}$$
so $x\trans{a} x'_c\in \lra_{QP + RP}$ and $(x',x'_c)\in S$. Similarly, we can prove that if $(x,y)\in S$ and $y\trans{a} y'$ then there exists $x'$ such that $(x'y')\in S$ and $x\trans{a} x'$. Hence $S$ is a bisimulation.
\item Left unfold \ref{eq:unfold}: Let $x_*$ be a state in $P^*$ and $x$ the corresponding state in the unfolded version $(1 + PP^*)$ i.e. $x$ is considered as a state of $P$. The rooted version of relation $S = id_{P^*}\cup\{(x_*,x)\}$ is a rooted $\eta$-bisimulation from $P^*$ to $1 + PP^*$.
\item Left induction \ref{hf:linduction}: as in (10), the proof is again similar to \cite{Rab11} because rooted $\eta$-simulation are stable by union.
\end{itemize}
\end{proof}

\begin{proposition}\label{apro:parallel-algebra}
$(\aut, +,\cdot,\pr{A} ,1 )$ satisfies equations (\ref{eq:par-assoc}- \ref{eq:exchange-law}) modulo rooted $\eta$-simulation equivalence for any set of synchronisable actions $A\subseteq\Sigma$ (i.e. no probabilistic actions).
\end{proposition}

\begin{proof}
$1\pr{A} 1 = 1 $ follows directly from the definition of $\pr{A} $ and the simulation used for the commutativity is $\{((x,y),(y,x))\ |\ x\in P\textrm{ and } y\in Q\}$.
\begin{itemize}
\item[(\ref{eq:par-assoc})] For associativity, we show that if $(x,(y,z))\trans{a} (x',(y',z'))\in\lra_{P\pr{A} (Q\pr{A} R)}$ then $((x,y),z)\trans{a} ((x',y'),z')\in\lra_{(P\pr{A} Q)\pr{A} R}$.
\begin{itemize}
\item If $a\notin A$, then
\begin{itemize}
\item $x\trans{a} x'$ and $y = y', z = z'$. So $(x,y)\trans{a} (x',y)\in \lra_{P\pr{A} Q}$ and hence $((x,y),z)\trans{a} ((x',y),z)\in \lra_{(P\pr{A} Q)\pr{A} R}$ because $a\notin A$.

\item or $x = x'$ and $(y,z)\trans{a} (y',z')$. Since $a\notin A$:
	\begin{itemize}
	\item $y\trans{a} y'$ and $z = z'$, hence $((x,y),z)\trans{a} ((x,y'),z)\in \lra_{(P\pr{A} Q)\pr{A} R}$,
	\item or $y = y'$ and $z\trans{a} z'$ and hence $((x,y),z)\trans{a} ((x,y),z')\in \lra_{(P\pr{A} Q)\pr{A} R}$.
	\end{itemize}
\end{itemize}
\item If $a\in A$, then $x\trans{a} x'$ and $(y,z)\trans{a} (y',z')$. Since $a$ is again synchronised in $Q\pr{A} R$, 
$y\trans{a} y'$ and $z\trans{a} z'$. So $(x,y)\trans{a} (x',y')\in \lra_{P\pr{A} Q}$ and hence $((x,y),z)\trans{a} ((x',y'),z')\in \lra_{(P\pr{A} Q)\pr{A} R}$.
\end{itemize}
Since $\pr{A} $ is commutative, we deduce that $\lra_{(P\pr{A} Q)\pr{A} R} = \lra_{P\pr{A} (Q\pr{A} R)}$ so the identity relations could again be used for the simulation.
\item[(\ref{eq:par-dist})] To prove monotonicity, we consider the relation $S:P \pr{A} Q + {P_c} \pr{A} R\to P\pr{A} (Q + R)$ as in the case of multiplication i.e. $S = \{((x,y),(x,y)),((x_c,y),(x,y))\ |\ x\in P\wedge y\in Q\cup R\}$ and $x_c$ is the copy of the state $x\in P$ in $P_C$. Let $((x_c,y),(x,y))\in S$ (the case $((x,y),(x,y))\in S$ is easier and can be handled in the same way) and $(x_c,y)\trans{a} (x'_c,y')\in\lra_{P\pr{A} Q+P\pr{A} R}$. By definition of $+$, that transition belongs to $\lra_{P\pr{A} Q}$ or $\lra_{P\pr{A} R}$. Since the first component is a copy of $x$, we have $(x_c,y)\trans{a} (x'_c,y')\in\lra_{P\pr{A} R}$ that is $y,y'\in R$.
\begin{itemize}
\item if $a\notin A$, then 
	\begin{itemize}
	\item $x_c\trans{a} x'_c$ and $y = y'$, so $x\trans{a} x'\in\lra_{P}$ and hence $(x,y)\trans{a} (x',y)\in\lra_{P\pr{A} (Q+R)}$ and $((x'_c,y),(x',y))\in S$ by definition of $S$.
	\item or $x_c= x'_c$ and $y\trans{a} y'$, so $(x,y)\trans{a} (x,y')\in\lra_{P\pr{A} (Q+R)}$ and $((x,y'),(x_c,y'))\in S$.
	\end{itemize}
\item if $a\in A$, then $x_c\trans{a} x'_c$ and $y\trans{a} y\in\lra_{R}'$. So $x\trans{a} x'\in\lra_{P}$ and hence $(x,y)\trans{a} (x',y')\in\lra_{P\pr{A} (Q+R)}$ and $((x'_c,y'),(x',y'))\in S$.
\end{itemize}
\item[(\ref{eq:exchange-law})] Firstly notice that the set of states of $(P\|Q)(P'\|Q')$ (where the frame $A$ of the concurrency operator is left implicit) is a subset of $(P\times Q)\cup(P'\times Q')$ which is in turn a subset of $(P\cup P')\times (Q\cup Q')$. Hence we consider the injection $id$ of the former set to the later one and the relation defined in Figure \ref{fig:sim-exchangelaw}
\begin{figure*}\label{fig:sim-exchangelaw}
\begin{eqnarray}
S = id & \cup & \{((x',i),(x',y))\ |\ y\in F_Q\textrm{ and } i\trans{a} x'\in \lra_{P'} \textrm{ and } i\in Q'\}\nonumber\\
	& \cup &\{((i,y'),(x,y'))\ |\ x\in F_P\textrm{ and } i\trans{a} y'\in \lra_{Q'} \textrm{ and } i\in P'\}\nonumber
\end{eqnarray}
\caption{Construction of the simulation to prove the interchange law.}
\end{figure*}
We show that $S$ is a simulation in our sense.
\begin{itemize}
\item[-] Since $(i,i) = i$ is related to itself. In particular, $S$ is rooted because $x'\neq i$ in the second set in the definition of $S$ (resp. for the third set) and HCI.
\item[-] Let $(x,y)\in (P\|Q)(P'\|Q')$ such that $(x,y)\trans{a} (x',y')$. We have the following cases.
\begin{itemize}
\item The transition is in $\lra_{P\|Q}$, in which case $(x,y),(x',y')\in P\times Q$. 
\begin{itemize}
\item if $a\notin A$ then $x\trans{a} x'\in\lra_P$ and $y = y'$ or $x = x'$ and $y\trans{a} y'\in\lra_Q$. By definition of the sequential composition again, these transitions belong to $\lra_{PP'}$ or $\lra_{QQ'}$ respectively. Hence $(x,y)\trans{a} (x',y')\in\lra_{PP'\|QQ'}$.
\item if $a\in A$ then $x\trans{a} x'\in\lra_P$ and $y\trans{a} y'\in\lra_Q$. As in the previous case, the considered transition exists in $PP'\|QQ'$
\end{itemize}
\item The transition is in $\lra_{P'\|Q'-\{(i,i)\}}$. This case is similar to the previous one because because $x\neq i$ and $y\neq i$ as states of $P'$ and $Q'$ respectively.
\item It is a linking transition i.e. $(x,y)\in F_{P\|Q}$ and $(i,i)\trans{a} (x',y')\in\lra_{P'\|Q'}$. Then $x\in F_P$ and $y\in F_Q$ and we have two cases:
\begin{itemize}
\item if $a\notin A$, then $i\trans{a} x'\in\lra_{P'}$ and $y' = i$ or $x' = i$ and $i\trans{a} y'\in\lra_{Q'}$. In the first case, the definition of $S$ implies that $((x',y'),(x',y))\in S$ and since $a\notin A$, we have $(x,y)\trans{a} (x',y)\in\lra_{PP'\|QQ'}$. Similarly for the other case.
\item if $a\in A$, then $i\trans{a} x'\in\lra_{P'}$ and $i\trans{a} y'\in\lra_{Q'}$. Then $x\trans{a} x'\in\lra_{PP'}$ and $y\trans{a} y'\in\lra_{QQ'}$. Hence $(x,y)\trans{a} (x',y') \in\lra_{PP'\|QQ'}$.
\end{itemize}
\item Let $((x',i),(x',y))\in S$ as in the above definition of $S$ and $(x',i)\trans{a} (x'',y'')\in\lra_{P'\|Q'}$.
\begin{itemize}
\item If $a\notin A$, then $x'\trans{a} x''\in\lra_{P'}$ and $y''=i$ or $x' = x''$ and $i\trans{a} y''\in\lra_{Q'}$. In the first case, $(x',y)\trans{a} (x'',y)\in\lra_{PP'\|QQ'}$ because $a\notin A$ and $(x'',y)\in S$ because $y'' = i$. In the second case, $y\trans{a} y''\in\lra_{QQ'}$ and hence $(x',y)\trans{a} (x'',y'')\in\lra_{PP'\|QQ'}$ because $x' = x''$.
\item If $a\in A$, then $x'\trans{a} x''\in\lra_{P'}$ and $i\trans{a} y''\in\lra_{Q'}$. By definition of sequential composition, $y\trans{a} y''\in\lra{QQ'}$ and since $a\in A$, $(x',y)\trans{a} (x'',y'')\in\lra_{PP'\|QQ'}$.
\end{itemize}
\item The case $((i,y'),(x,y'))\in S$ is similar.
\end{itemize}
\item[-] Let $((x,y),(u,v))\in S$ such that $(x,y)\in F_{(P\|Q)(P'\|Q')}$. That is, 
$$x\in (F_{P'}-\{i\})\cup o_{P'\|Q'}((i,i))F_{P}\subseteq (F_{P'}-\{i\})\cup o_{P'}(i)F_{P}$$
and similarly for $y$. Hence, if tuple belongs to $id$ then we are done. Assume $i\trans{a} x\in\lra_{P'}$ and $y = i$ (the other case is proved in exactly the same way), then $u = x$ and $v\in F_Q$. Since $(x,i)$ is a final state, we have $F_{QQ'} = (F_{Q'}-\{i\})\cup F_{Q}$ and $x'\in F_{P'}-\{i\}$. Hence $(u,v)\in F_{PP'}\times F_{QQ'}$.
\end{itemize}
Finally, since simulation preserves reachability, the reachable part of $(P\|Q)(P'\|Q')$ is simulated by the reachable part of $PP'\|QQ'$.
\end{itemize}
\end{proof}

\begin{proposition}\label{pro:mult-cont}
The sequential composition is (conditionally) continuous from the left and the right in $\aut_f$. That is, if $(P_i)_i$ is a $\leq$-directed set of finite automata with limit $P$ then $\sup_i P_iK = PK$ and $\sup_i KP_i = KP$.
\end{proposition}

We denote $\aut_f$ the set of finite automata satisfying the reachability and initiality conditions. $\aut_f$ is a subalgebra of $\aut$. The proof is similar to our proof in~\cite{Rab11}. The only difference is from the manipulation of $0$ (because we do not have $x0 = 0$ in this setting) and hence Proposition~\ref{pro:mult-cont} is a generalised version of the continuity in~\cite{Rab11}.

\begin{proof}
We first define a notion of residuation on $\aut_f$. For automata $P$ and $Q$ we define the automaton $P/Q$ with initial state $i_{P/Q} = i_P$, final states $F_{P/Q} = \{x\in P\ |\ Q\leq P_x\}$, where $P_x$ is constructed from $P$ by making its initial state into $x$. We make the resulting automaton reachable by discarding all states not reachable from $x$. Notice that $P_x$ does not necessarily satisfy HCI. In this case, we unfold each transition from $x$ once and isolate $x$ but keeping a disjoint copy of it to make sure that the resulting automata is bisimulation equivalent to the non-rooted version. 

We now show that $RQ\leq P\textrm{ iff } R\leq P/Q$.  Assume $S$ is a simulation from $RQ$ to $P$. That means $S$ generates a simulation from $Q$ to $P_x$ for some $x$. It follows from the definition of $P/Q$ that $S$ generates a simulation from $R$ to $P/Q$, since the state $x$ become final state of $G/H$ and they are images of the final states of $K$ under the simulation generated by $S$.

For the converse direction, suppose that $S$ is a simulation from $R$ to $P/Q$. By Theorem \ref{pro:pka-soundness}, multiplication is isotone, hence $RQ\leq (P/Q)Q$, and it remains to show that $(P/Q)Q\leq P$.

First, if $F_{P/Q}$ is empty and then $R$ has no final state either and $RQ = R$ by definition of sequential composition. Hence $RQ = R\leq P/Q\leq P$.

Assume $F_{P/Q}$ is not empty and let $S'$ be a simulation from $RQ$ to $(P/Q)Q$. By construction of $P/Q$, we know that there exists a simulation $S_x$ from $Q$ to $P_x$ for all final states $x\in F_{P/Q}$. Moreover, there is a relation $T:P/Q\rightarrow P$ satisfying all properties of simulation except the final state property, namely a restriction of the identity relation $id_P$. We can show that $T' =  (\cup_x S_x)\cup T$ is indeed a simulation from $(P/Q)Q$ to $P$ and $S'\circ T'$ is a simulation from $RQ$ to $P$.

It then follows from general properties of Galois connections that $(\cdot H)$ is (conditionally) completely additive, hence right continuous.

It remains to show left continuity. Let $(Q_i)_i$ be a directed set of automata such that $\sup_iQ_i = Q$ and let $P$ be any automaton. Then $\sup_i(PQ_i) \leq PQ$ because multiplication is monotone and it remains to show $PQ \leq \sup_i(PQ_i)$. Let us assume that $\sup_i(PQ_i)\leq R$. We will show that $PQ\leq R$.

By definition of supremum, $PQ_i\leq R$ for all $i$, hence there is a set of states $X_i = \{x\in R\ |\ Q_i\leq R_x\}$, that is, the set of all those states in $R$ from which $Q_i$ is simulated. Obviously, $X_i\subseteq X_j$ if $Q_j\leq Q_i$ in the directed collection.  But since $R\in\aut_f$ has only finitely many states, there must be a minimal set $X$ in that directed set such that all $Q_i$ are simulated by $R_x$ for some $x\in X$.  Therefore $Q = \sup_iQ_i\leq R_x$ for all $x\in X$.  There exists a simulation $S_X: PQ_i\rightarrow R$ for some $i$ such that the residual automaton $R/Q_i$ has precisely $X$ as its set of final states.  We can thus take the union of $S_X$ restricted to $P$ with all
simulations yielding $Q\leq R_x$ for all $x\in X$ and verify that
this is indeed a simulation of $PQ$ to $R$.
\end{proof}

We denote $L(P) = \{t \ |\ t \textrm{ is a tree and } t\leq P\}$ the tree language associated to $P$. We have 

\begin{lemma}\label{lem:tree-lang}
$P\leq Q$ iff $L(P)\subseteq L(Q)$.
\end{lemma}

A specialized version of this theorem could be found in~\cite{Rab11}. In this paper, we prove it for our rooted $\eta$-simulation.

\begin{proof}
By transitivity of simulation, we have $P\leq Q$ implies $L(P)\subseteq L(Q)$ so it suffices to show the converse.

Let $L(P)\subseteq L(Q)$ and consider the relation $S:P\rightarrow Q$ such that $(x,y)\in S$ iff $L(P_x)\subseteq L(Q_y)$, where $P_x$ is the automata constructed from $P$ with initial state $x$ as in the previous proof. We now show that the rooted version of $S$ is a simulation.
\begin{itemize}
\item Since $L(P)\subseteq L(Q)$, we have $(i_P,i_Q)\in S$.
\item Let $(x,y)\in S$ and $x\in F_P$, then $1\in L(P_x)\subseteq L(Q_y)$. Hence $y\in F_Q$.
\item Let $(x,y)\in S$, $L(P_x)\subseteq L(Q_y)$ and $x\trans{a} x'$ be a transition of $P$. There are two cases:
\begin{itemize}
\item[a)] $a$ is internal: for any tree $t$, $at\leq t$. Hence $L(P_{x'})\subseteq L(Q_y)$ i.e. $(x',y)\in S$.
\item[b)] $a$ is external: assume for a contradiction that for each $y'_i\in Q$ such that $y\Rightarrow y_1\trans{a} y'_i$, there exists $t_i\in L(P_{x'})$ such that $t_i\notin L(Q_{y'_i})$. Since $Q\in \aut_f$, there are only finitely many such $y'_i$. By definition of $\eta$-simulation, $a(\sum_i t_i)\in L(P_x)$ and from it follows from the hypothesis that $a(\sum_i t_i)\in L(Q_y)$ i.e. $a(\sum_i t_i)\leq Q_y$. It follows from the definition of $\eta$-simulation that there exists $y'_j$ such that $y\Rightarrow y_1\trans{a} y'_j$ and $\sum_i t_i\leq y'_j$ which implies $t_j\leq\sum_it_i\leq y_j$, a contradiction.
\end{itemize}
\item Making the relation $S$ rooted does not affect the well-definedness of $S$ as a simulation because the automata $P,Q$ are rooted.
\end{itemize} 
\end{proof}

\begin{proposition}\label{pro:par-continuous}
The concurrency operator $\|$ is (conditionally) continuous in $\aut_f$.
\end{proposition}

\begin{proof}
We need to show that for any $\leq$-directed sequence $(Q_i)_i\subseteq \aut_f$ such that $\sup_i Q_i = Q$, we have $\sup_i (P\| Q_i) = P\| Q$, where the frame is left implicit. 

Firstly, we show that $$L(P\| Q) = \downarrow(L(P)\| L(Q)) = \downarrow\{t\| t' \ |\ t\in L(P)\wedge t'\in L(Q)\}$$
where $\downarrow X$ is the down closure of $X$.
Since $\| $ is monotone, $t\|t'\in L(P\| Q)$. Conversely, let $t\in L(P\pr{A} Q)$. By unfolding $P$ and $Q$ up to the depth of $t$, we can find two tree $t_P,t_Q$ such that $t\leq {t_P} \| t_Q $ and hence $t\in \downarrow (L(P)\| L(Q))$.

Secondly, we have 
\begin{eqnarray}
L(P\| Q) & = & \downarrow\{ t\| t'\ |\ t\in P\wedge t'\in \cup_iL(Q_i)\}\nonumber\\
& = & \downarrow \cup_i\{t\| t' \ |\ t\in P\wedge t'\in L(Q_i)\}\nonumber\\
& = & \cup_i\downarrow\{ t \|t'\ |\ t\in P\wedge t'\in L(Q_i)\}\nonumber\\
& = & \cup_iL(P\|Q_i)\nonumber
\end{eqnarray} 
 and directedness ensures that $L(Q) = \cup_iL(Q_i)$~\footnote{$\cup_iL_i\subseteq Q$ is obvious and the converse could be proven by showing that if $t\in L(Q)$ and $t\notin L(Q_i)$ for every $i$, then there exists $Q'$ constructed from $Q$ ``minus" some part of $t$ such that $Q_i\leq Q'< Q$ for any $i$.}. Therefore, Lemma \ref{lem:tree-lang} ensures that $P\|Q = \sup_i(P\|Q_i)$.
\end{proof}

\begin{remark}\label{rem:remark-o}
The following remarks ensures the existence of an $o$ function satisfying the properties listed in Section \ref{sec:algebraic-testing}
\begin{enumerate}
\item The axiom $\tau\tau=\tau$ ensures that $o(x)\in\{0,\tau,1\}$ for any $x\in T_\Sigma$. If $o(x) = 0$ then $x$ will \textit{never terminate successfully}. If $o(x) = 1$,  then $x$ \textit{may terminate successfully} without the execution of any action~\footnote{A rigorous proof of this fact could be done by induction on the structure of $x$.}, and if $o(x) = \tau$ then $x$ \textit{may terminate successfully} after the execution of some action.
\item The interpretation of $o$ in the concrete model respects simulation. In fact, let $P,Q$ be the automata representing some terms in $T_\Sigma$ and $S:P\rightarrow Q$ be a simulation. After replacing each action in $P,Q$ by $\tau$, $S$ remains a simulation by Propriety (a) of Definition \ref{df:sim}. Therefore 
\begin{itemize}
\item[-] if $o(P) = 1$ then the initial state of $P$ is final and so is the initial state of $Q$, 
\item[-] if $o(P) = \tau$ then the initial state of $P$ leads to some final state and so is the initial state of $Q$ i.e. $1\leq o(Q)$,
\item[-] if $o(P) = 0$ then we are done,
\end{itemize}
and in all three cases $o(P)\leq o(Q)$.
Hence, it is safe to assume that $o$ is well defined on $T_\Sigma$ modulo the axioms of weak concurrent Kleene algebra. In particular, $o$ is monotonic with respect to the restriction of the natural order of the algebra on $I$.
\item The last property $o(x\|y)\leq o(x)o(y)$ is in general a strict inequality. For instance, if $a,b$ are synchronised actions then $o(a\|b) = o(0) = 0$ but $o(a)o(b) = \tau\tau = \tau$.
\end{enumerate}
\end{remark}

\begin{proposition}\label{apro:may-language}
In $\aut$, $\may{} $ reduces to language equivalence.
\end{proposition}
Remind that we assume there is only one non-trivial internal action, namely $\tau$, and it satisfies $\tau\tau = \tau$. 
\begin{proof}
Firstly, the language of the automata associated to $x$ is given by
\begin{eqnarray}
Tr(x) & = & \{t\ |\ t \textrm{ is linear, loop-free, has only $\tau$ as} \nonumber \\
	& &\quad \textrm{non-synchronised action and } t\leq_\tau x\}\nonumber
\end{eqnarray}
where $x\leq_\tau y$ if there is a simulation between the automata represented by $x$ and $y$ such that all non-synchronised actions are replaced by $\tau$. This ensures for instance that $Tr(\flip) = \{\tau\}$.

Remind that $Tr(x\| y) = Tr(x) \cap Tr(y)$~\footnote{A small difference from CSP is that we consider only words terminating to final states but since $F_{x\|y} = F_x\times F_y$, we are safe to use  most of the general properties found in CSP such as $$Tr(x\| y) = \{t \ |\ t_{|x}\in Tr(x) \wedge t_{|y} \in Tr(y) \}$$.} because elements of $Tr(x)$ are of the form $w\tau$ or $w$ (modulo the equivalence from $\leq_\tau$) where $w$ is a word formed of synchronised actions only.

For the direct implication, assume $x\may{} y$ and let $t\in Tr(x)$. Then $o(x\|t)\neq 0$ and since $x\|t\leq y\|t$, we have $o(y\| t)\neq 0$. Since $t$ has synchronised actions only (or possibly ends with $\tau$) and $o(y\| t)\neq 0$, then $t\in Tr(y)$ that is $Tr(x)\subseteq Tr(y)$. 

Conversely, let $Tr(x)\subseteq Tr(y)$ and $z\in T_\Sigma$. $Tr(x\|z) = Tr(x)\cap Tr(z)\subseteq Tr(y)\cap Tr(z) = Tr(y\|z)$. So if $o(y\|z) = 0$ then $y\| z$ has no final state and hence $Tr(y\| z) = \{0\}$. Hence $Tr(x\|z) = \{0\}$ i.e. $x\|z$ has no final state that is $o(x\| z) = 0$. 
\end{proof}

\section{Specification of Rabin's Protocol.}

Remind that $P(\alpha,k)$ is the specification of a tourist in from of the door $\alpha\in\{m,c\}$ and has $k$ written on his notepad (Figure \ref{fig:app-rabin}).
\begin{figure}
$$\xymatrix{
&P(\alpha,k)\ar[d]_{\alpha?K}&&\\
&\ar[dl]_{[K=here]}\ar[dr]^{[K\neq here]}&&\\
\bullet	& \ar[l]^{\alpha!here}& \ar[l]_{[k>K]}\ar[dl]^{[k<K]}\ar[d]^{[k=K]}&\\
	& \ar[dl]^{[k:=K]} 			&\ar[d]^{\flip_{1/2}}& \\
\ar[d]_{\alpha!k}& &\ar[dl]_{\tau_h}\ar[dr]^{\tau_t} &\\
\ar[d]_{[\alpha := \underline\alpha]}& \ar[dr]_{[k := K+2]}& &\ar[dl]^{[k:=\overline{K+2}]}\\
\circ&&\ar[d]^{\alpha!k}&\\
& &\ar[d]^{[\alpha := \underline{\alpha}]} &\\
& &\circ &
}
$$
\caption{Interpretation of $P(\alpha,k)$ in term of automata with imiplicit probability.}
\label{fig:app-rabin}
\end{figure}
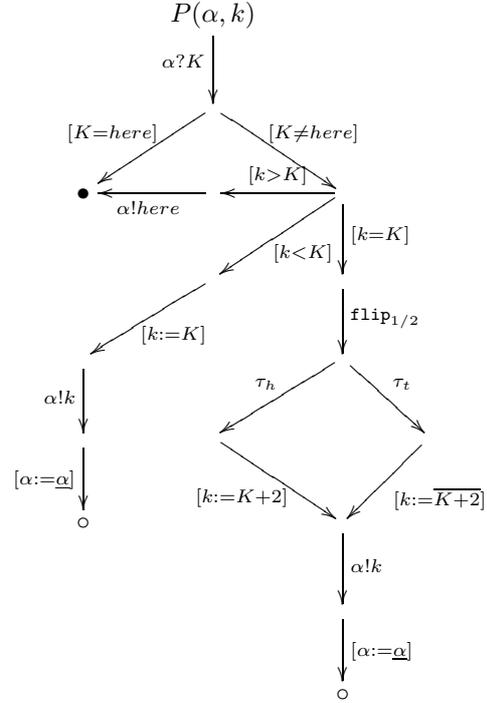

Any action of the form $[a]$ are considered internal. The symbol $\circ$ denotes final states and $\bullet$ is a deadlock state. In this protocol, deadlock state is used to specify that the tourist has come to a decision and the common place would be the value of $\alpha$ when the deadlock state is reached.

\begin{theorem}
In the concrete model, the specification of Rabin's protocol satisfies $$S = [((P+Q)\|M)^*((P+Q)\|C)^*]^*$$
\end{theorem}
Firstly, notice that if $\cdot$ is (conditionally) continuous then $x^* = \sup_{n\in\mathbb{N}}(1+x)^n$. The proof relies on the fact that $f_x^n(0) = (1 + x)^n$ where $f_x(y) = 1 + x\cdot y$ and the result follows by taking the limit.

\begin{proof}
The above property allows us to express $x^*$ as the limit of finite iterations of $x$ interleaved with successful termination. We have 
\begin{eqnarray}
(P + Q)^*\|(M+C)^* & = & \sup_m(1 + P + Q)^m\|\sup_n(1 + M + C)^n\nonumber\\
& = & \sup_m\sup_n\left[(1 + P+Q)^m\|(1 + M + C)^n\right]\nonumber
\end{eqnarray}

The processes $P$ and $Q$ are essentially delimited by $\alpha?K$ and $\alpha!K$ which ensures the following properties of the system
\begin{eqnarray}
X \cdot A\|Y\cdot B & = & [X\|Y]\cdot [A\|B]\label{eq:p1}\\
X\cdot A\|1 & = & 0 \label{eq:p2}\\
Y\cdot B\|1 & = & 0\label{eq:p2'}
\end{eqnarray}
for every processes $A,B$ and where $X = P+Q$ is the collection of tourists and $Y = M+C$ is the collection of places.

In particular, 
$$1\|(1 + X)^n = 1\|(1 + X)^{n-1} + 1\|X\cdot(1 + X)^n = 1\|(1 + X)^{n-1}$$
and by induction, since $1\|1 = 1$,
\begin{equation}\label{eq:p3}
1\|(1 + X)^n = 1
\end{equation}
for every $n\in\mathbb{N}$. Similarly, $1\|(1 + Y)^n = 1$.

On the other hand, let us denote $T_{m,n} = (1 + X)^m\|(1 + Y)^n$, then 
\begin{eqnarray}
T_{m,n}& = & \left[(1 +X)^{m-1} + X\cdot(1+X)^{m-1}\right]\|\nonumber\\
 & & \qquad \left[(1 + Y)^{n-1} + Y\cdot (1+Y)^{n-1}\right]\nonumber\\
 & = & T_{m-1,n-1} + X\cdot(1 + X)^{m-1}\|(1 + Y)^{n-1} + \nonumber \\
& & \qquad(1+X)^{m-1}\|Y\cdot (1+Y)^{m-1} + \nonumber \\
& & \qquad [X\|Y]\cdot T_{m-1,n-1}\nonumber\\
& = & (1 + X\|Y)\cdot T_{m-1,n-1} + U_{m-1,n-1} + \nonumber\\
& &\qquad V_{m-1,n-1}\nonumber
\end{eqnarray}
where 
\begin{eqnarray}
U_{m-1,n-1} & = & X\cdot(1 + X)^{m-1}\|[(1 + Y)^{n-1}\nonumber\\
& = & U_{m-1,n-2} + [X\|Y]\cdot T_{m-1,n-2}\nonumber\\
& = & U_{m-1,n-3} +  [X\|Y]\cdot T_{m-1,n-3} + \nonumber\\
& &\qquad [X\|Y]\cdot T_{m-1,n-3}\nonumber
\end{eqnarray}

Since the sequence $(1+Y)^n$ is monotone, $T_{m,n}\leq T_{m,n'}$ for every $n\leq n'$ and therefore $U_{m-1,n-1}\leq U_{m-1,0} + [X\|Y]\cdot T_{m-1,n-1}$. But Property \ref{eq:p2} implies that $U_{m-1,0} = 0$.

Similarly, $V_{m-1,n-1} \leq [X\|Y]\cdot T_{m-1,n-1}$. Hence 
$$T_{m,n} = (1 + [X\|Y])\cdot T_{m-1,n-1}.$$ 
By induction, we show that 
$$T_{m,n} = (1 + [X\|Y])^{\inf(m,n)}$$
because $T_{0,n} = T_{m,0} = 1$ by Equation \ref{eq:p3}.

Finally, we have 
\begin{eqnarray}
X^*\|Y^* & = & \sup_m\sup_n (1+X)^m\|(1+Y)^n\nonumber\\
& = & \sup_n (1 + [X\|Y])^n\nonumber\\
& = & (X\|Y)^*\nonumber
\end{eqnarray}
\end{proof}

\end{document}